\documentclass{jfm}
\usepackage{graphicx}
\usepackage{epstopdf, epsfig}
\usepackage{natbib}
\usepackage{hyperref}
\usepackage{physics}
\usepackage{amsmath}
\usepackage{bm}
\usepackage{cancel}
\usepackage{booktabs,multirow}
\usepackage{ragged2e}
\usepackage{xcolor}

\newtheorem{corollary}{Corollary}
\newtheorem{hyp}{Hypothesis}

\newcommand{\uvec}{\hat{\bm{e}}}
\newcommand{\vq}{\bm{v}_0}

\newcommand{\iu}{\mathrm{i}}
\newcommand{\sgn}{\mathrm{sgn}}

\shorttitle{Selection rules for triadic resonance in columnar vortices}
\shortauthor{J. Wang, S. Lee and P. S. Marcus}

\title{Perturbation analysis of triadic resonance in columnar vortices: selection rules and the roles of external forcing and critical layers}

\author{Jinge Wang\aff{1}, Sangjoon Lee\aff{2} \and Philip S. Marcus\aff{1} \corresp{\email{pmarcus@me.berkeley.edu}}}

\affiliation{\aff{1}Department of Mechanical Engineering, University of California, Berkeley, CA 94720, USA
\aff{2}School of Mechanical and Aerospace Engineering, Nanyang Technological University, Singapore 639798, Republic of Singapore}

\begin{document}

\maketitle

\begin{abstract}
The remarkable robustness of columnar vortices suggests the existence of fundamental constraints that prevent spontaneous disintegration. In this work, we investigate the weakly nonlinear stability of such flows, demonstrating that the triadic resonance of wave modes is governed by a set of hydrodynamic ``selection rules''. By employing a multi-scale perturbation analysis, we prove that resonant interactions between smooth neutral modes, specifically regular Kelvin waves and discrete critical layer modes with passive singularities, are strictly conservative and confined to the Manley--Rowe relations. Using wave pseudoenergy within a large-$k$ WKBJ framework, we show that these rules topologically prohibit intrinsic instability, analogous to the forbidden transitions of quantum mechanics. Consequently, the breakdown of a columnar vortex requires a specific symmetry-breaking mechanism to overcome this barrier. We identify two distinct pathways: (1) \textit{Parametric instability}, a limiting case where one mode is maintained externally. By generalising beyond the specific spatial and temporal assumptions of classical studies (e.g., elliptical instability) and leveraging a tuning method based on non-degenerate perturbation theory, our framework admits arbitrary driving frequencies and identifies new instability configurations involving discrete critical layer modes. (2) \textit{Active critical layers}, where an embedded wave-mean resonance enables the direct, non-conservative extraction of mean-flow energy. These findings provide theoretical guidance for flow control, suggesting that aircraft wake vortex mitigation requires either tuned external forcing or the excitation of critical layers (e.g., via thermal stratification) to trigger the forbidden transitions.
\end{abstract}

\section{Introduction}
The ubiquity of columnar vortices in nature and engineering, ranging from the enduring Great Red Spot of Jupiter to the intense swirl of tornadoes and the hazardous wake turbulence of aircrafts, underscores the necessity of understanding their stability mechanisms. Historically, the study of fluid stability originated from linear stability theory, dating back to Lord \citet{Kelvin_1880} and his seminal analysis of the Rankine vortex. While this framework has successfully classified the discrete spectrum of regular modes (Kelvin waves) and accurately predicts the initial growth of disturbances in flows subject to centrifugal or Rayleigh--Taylor instabilities \citep{Rayleigh_1916, Drazin_2004}, it faces significant limitations when applied to robust, long-lived structures. For vortices such as aircraft wake vortices (AWV), 
linear theory in the inviscid limit implies neutral stability due to the temporal symmetry of the Euler equations. Consequently, linear analysis alone cannot account for the evolution and eventual breakdown of these coherent structures. To elucidate the mechanisms driving their destabilisation, it is necessary to transcend the linear approximation and consider the effects of higher-order terms \citep{Schmid_Marsden_Sirovich_2001}. These nonlinear interactions form the basis of weakly nonlinear theory.

In the context of weakly nonlinear theory, significant attention has been focused on vortices subjected to sustained external deformation \citep[see][]{kerswell_2002}. Classical examples include the elliptical instability driven by an external strain field \citep{mcewan_1970, saffman_1975, tsai_widnall_1976, dizes_1999, kerswell_1999, kerswell_2002, dizes_2005, dizes_2007, feys_maslowe_2016, mora_2021}, the precessional instability inside a fluid cylinder \citep{mahalov_1993, Manasseh_1992, Manasseh_1994, Kobine_1995, meunier_2008, eloy_2011, meunier_2015, lopez_2018, meunier_2018}, and the curvature instability arising from the bending of the vortex filament \citep{fukumoto_2005, dizes_2017}. 
Mathematically, these phenomena can be classified as parametric instabilities which arise when an externally maintained wave mode acts as an active energy pump for a pair of disturbance waves. While prior studies have established a robust framework for forced vortices, they have largely focused on configurations where the imposed mode is stationary and translationally invariant along the vortex axis. These assumptions, though physically justified for modelling specific environments, leave a critical gap in our understanding of the broader stability landscape, particularly for isolated vortices where no such sustained external deformation is present.

This gap is particularly relevant to the mitigation of aircraft wake vortices. These long-lived, counter-rotating vortices are generated at the wingtips and pose a significant hazard to the following aircraft \citep{hallock_2018}. The lifecycle of an AWV system typically involves an initial phase where the two vortices are sufficiently separated to be treated as isolated. Over time, viscous diffusion expands the vortex cores until they begin to impose a mutual strain effect on each other. This mutual interaction triggers mechanisms such as the long-wave Crow instability \citep{crow_1970} or the short-wave elliptical instability \citep{saffman_1975, tsai_widnall_1976, leweke_williamson_1998}, eventually leading to reconnection and breakdown due to the cancellation of circulation. The engineering challenge lies in the fact that the initial isolated phase is remarkably robust. To accelerate the decay of the wake, it is desirable to identify intrinsic mechanisms that increase the growth of disturbances during the isolated phase, thereby allowing the vortices to reach sizes sufficient to trigger mutual interactions earlier.

To induce instability in this robust regime, specific spectral conditions must be satisfied. In the inviscid limit, the eigenmodes of a stable isolated vortex are strictly neutral, with eigenvalues confined to the imaginary axis. Consequently, the transition to instability requires a pair of neutral eigenvalues to collide, creating a geometric \textit{degeneracy} that allows them to bifurcate into growing and decaying modes, while preserving the underlying temporal symmetry. \textit{Triadic resonance}, a universal phenomenon that occurs in various dispersive media, including plasma physics, nonlinear optics, and fluid mechanics \citep{weiland_1977, craik_1986}, provides the simplest nonlinear wave coupling required to unfold this spectral degeneracy. It encompasses interactions where all three wave modes are free to evolve, as well as classical parametric instabilities where one wave is maintained. By allowing for arbitrary spatial and temporal parameters, this general triadic resonance framework transcends the restricted assumptions of previous studies of vortex stability to map a vastly expanded landscape of potential instabilities.

In many systems, triadic wave interactions obey a set of conservation laws known as the Manley--Rowe relations \citep{Manley_Rowe_1956}, which constrain the exchange of wave action within each resonant triad. 
Despite their utility in other fields, these relations have not been rigorously explored in the context of vortical flows. Specifically, it has not been established whether the triadic coupling of wave modes supported by a background vortex constitutes a conservative resonance, and consequently, there is a lack of definitive conclusions regarding which modes, or if indeed all modes, adhere to these conservation laws. Determining this requires a formal analysis to evaluate the energetic nature of the coupling. As discussed by \citet{Dodin_2008}, these relations can be derived deductively to offer a compact form of integrals for dynamical systems with resonances. Once derived, this formulation yields interaction coefficients (or pseudoenergies) that dictate the nature of the wave coupling. Crucially, depending on the signs of these coefficients, the solution can either be bounded periodic exchange of wave activity or unbounded simultaneous growth to infinity. The latter solution is also known as \textit{explosive instability} \citep[see][]{craik_adam_1979, becker_1993}, which represents a possible pathway toward accelerating vortex breakdown without external forcing.

\subsection{Hydrodynamic selection rules and the optical analogy}
To elucidate the physics of wave resonance in vortical flows, we draw a theoretical connection to nonlinear optics. Both fluid vortices and optical media support dispersive waves that can interact nonlinearly. In optical three-wave mixing, energy transfer is efficient only when the phase matching condition ($\Delta\bm{k} = 0$) is satisfied \citep[see][]{boyd_2008}. Similarly, hydrodynamic resonance requires the geometric matching of wavenumbers ($k,m$ for the axial and azimuthal wavenumbers) and temporal frequencies ($\omega$) among the triad members.

The distinction between the two systems lies in the energetic status of the background medium. In standard nonlinear optics involving a passive medium (such as a transparent crystal), the total energy of the interacting waves is conserved. Under these conditions, the Manley--Rowe relations constrain the system to a bounded manifold, resulting in oscillatory exchange between frequencies \citep[see][]{Pershan_1918}. In contrast, the rotating shear of the background vortex constitutes a vast energy reservoir. Superficially, this resembles an optical ``active medium'' (e.g., a laser gain medium) or a plasma, where the medium can pump energy into the waves to drive explosive growth.

However, a central contribution of this paper is the demonstration that the availability of this reservoir does not guarantee instability. We establish a set of hydrodynamic ``selection rules'' to govern these interactions. Drawing a parallel to quantum mechanics, where selection rules dictate allowable atomic transitions, we demonstrate that for \textit{smooth neutral modes}, the Hamiltonian structure of the Euler equations imposes strict constraints that confine the interaction to the bounded solution branch. Despite the presence of free energy in the shear, these rules prevent the resonant triad from extracting it. Consequently, the breakdown of a columnar vortex requires a symmetry-breaking mechanism to either access the explosive solution branch or violate the conservative constraints entirely.

We identify two such pathways. The first is \textit{external forcing}, which functions analogously to an optical parametric oscillator where an external pump wave actively injects energy into the system. This triggers a generalised parametric instability. Unlike classical studies typically limited to resonant configurations dictated by geometric constraints, our framework admits arbitrary driving frequencies and imposes no restrictions on the pumping wavenumbers. The second is the \textit{critical layer}, a flow singularity where the phase velocity of the wave matches the local flow velocity \citep{jacquin_2006, dizes_2005_2}. The critical layer functions similarly to optical resonant layers found in plasma absorption \citep{stix_1992}, allowing the wave to extract energy from the background shear in a non-conservative manner. This analogy frames our investigation: intrinsic wave-wave resonance among regular modes is energetically constrained and bounded (passive), whereas instability requires a mechanism to tap into energy reservoirs, either explicitly via external forcing or implicitly via a critical layer singularity.

\subsection{Overview}
The objective of this paper is to perform a multi-scale perturbation analysis of triadic resonance in generic columnar vortices and to establish the selection rules that distinguish between bounded and explosive solutions. The paper is organised as follows. In \S\ref{sec:formulation}, we formulate the disturbance equations using a poloidal-toroidal decomposition. In \S\ref{sec:multi-scale}, we derive the general amplitude equations for triadic resonance. We analyse the specific case of conservative resonance, showing that it obeys the Manley--Rowe relations and theoretically admits either bounded oscillatory exchange or simultaneous explosive growth. In \S\ref{sec:selection-rules}, we apply this framework to columnar vortices. We prove that triadic resonance among smooth neutral modes is always conservative. Using wave pseudoenergy within a large-$k$ WKBJ framework, we then demonstrate that these conservative triads are strictly confined to the bounded solution branch, effectively prohibiting intrinsic explosive instability for smooth modes. We then relax the isolation constraint and investigate parametric instability, employing non-degenerate perturbation theory to tune for arbitrary resonant triads. Finally, in \S\ref{sec:critical-layer}, we discuss the role of critical layer singularities as a necessary pathway to break the selection rules by allowing the wave to access the background shear energy.

\section{Base flow and disturbance equations}\label{sec:formulation}
\begin{figure} 
    \centering 
    \includegraphics[width=.75\textwidth]{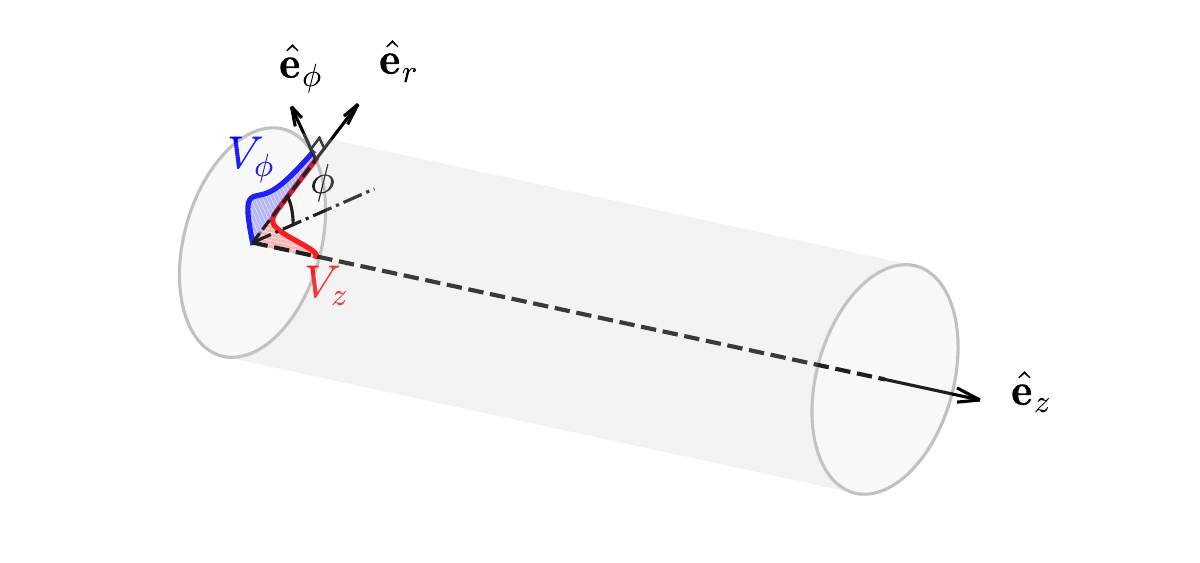}
    \caption{Schematic of the base flow configuration and the coordinate system. The transparent cylindrical surface is provided for visual reference only. The $z$-axis aligns with the vortex centreline. The depicted velocity profiles illustrate a generic columnar vortex possessing both an azimuthal swirl ($V_\phi$) and an axial jet ($V_z$). Note that both velocity components are exclusively functions of the radial coordinate $r$, and the base flow is strictly axisymmetric and axially homogeneous.} 
    \label{fig:schematic} 
\end{figure}

The base flow of interest is an isolated axisymmetric columnar vortex. As depicted by the schematic in figure~\ref{fig:schematic}, in a cylindrical coordinate system $(r,\phi,z)$ where the $z$-axis aligns with the vortex centreline, its velocity profile has the generic form 
\begin{equation}
    \vq = V_\phi(r)\uvec_\phi + V_z(r)\uvec_z,
    \label{eqn:col-vortex}
\end{equation}
where $V_\phi$ and $V_z$ are the azimuthal and axial velocity components, and the unit vectors along the $r$, $\phi$ and $z$ directions are denoted as $\uvec_r$, $\uvec_\phi$ and $\uvec_z$. Accordingly, the angular velocity $\Omega(r)$ and axial vorticity $\zeta(r)$ are expressed as
\refstepcounter{equation}
$$
\Omega(r) = \frac{V_\phi}{r}, \quad
\zeta(r) = \frac{1}{r}\dv{(rV_\phi)}{r}.
\eqno{(\theequation{a,b})}\label{eqn:base_zeta}
$$
Here, we assume that any provided physical quantities are non-dimensionalised with respect to a velocity scale $V_*$ and a length scale $r_*$, with the premise that they are well-defined. For instance, in the case of the Rankine vortex, $r_*$ may be taken as the radius of the vorticity patch, and $V_*$ as the maximum azimuthal velocity. Utilising these characteristic scales, the Reynolds number of the flow is defined as $\Rey \equiv r_* V_* / \nu$, where $\nu$ is the kinematic viscosity of the fluid.

In this study, we will assume that the Reynolds number is sufficiently high ($\Rey \gg 1$), as is typical for aeronautical applications \citep{dizes_2007}, so viscous diffusion is neglected. The columnar vortex, as defined by (\ref{eqn:col-vortex}), is an equilibrium solution to the incompressible Euler equations regardless of the profiles of $V_\phi(r)$ and $V_z(r)$. The multiple-scale perturbation analysis presented in the subsequent sections will be developed for arbitrary $\vq$ profiles. However, for illustrative purposes and numerical computations, we will refer to the specific case of the Lamb--Oseen vortex, given by
\begin{equation}
    \vq = \frac{1-e^{-r^2}}{r}\uvec_\phi.
\end{equation}



We consider a small but finite-amplitude disturbance velocity field, $\bm{u}$, superposed on the base flow (\ref{eqn:col-vortex}). The equations of motion for the total velocity field, $\bm{v} = \vq + \bm{u}$, are the incompressible Euler equations, here written in the rotation form:
\begin{subequations}\label{eqn:euler-eqns}
\begin{gather}
    \div{\bm{v}} = 0, \label{eqn:euler-eqns-a} \\
    \pdv{\bm{v}}{t} = \bm{v}\cross(\curl{\bm{v}}) - \grad\varphi. \label{eqn:euler-eqns-b}
\end{gather}
\end{subequations}
where $\varphi = (\bm{v} \cdot \bm{v}) / 2 + p / \rho$ with $p$ being the associated pressure field and $\rho$ being the density of the fluid. 
We require the disturbance $\bm{u}$ to be analytic at $r = 0$ and to decay algebraically or exponentially as $r\rightarrow\infty$ to maintain finite kinetic energy. Detailed discussions on the physical and accurate attainment of these conditions can be found in \citet{matsushima_marcus_1997, lee_marcus, Lee2024}.

The disturbance equations are obtained by cancelling out the equilibrium of $\vq$ from (\ref{eqn:euler-eqns-b}), which results into a system of four equations with four primitive variables: three components of $\bm{u}$, and $\varphi$. Instead of dealing with the primitive form, we note that the incompressibility condition (\ref{eqn:euler-eqns-a}) reduces to $\div{\bm{u}} = 0$, so the disturbance velocity field can be decomposed into poloidal and toroidal components:
\begin{equation}
    \bm{u} = \curl{(\psi\uvec_z)} + \curl{\left(\curl{(\chi\uvec_z)}\right)},
    \label{eqn:pt-decomp}
\end{equation}
where $\psi$ and $\chi$ are the toroidal and poloidal streamfunctions. To avoid gauge arbitrariness in (\ref{eqn:pt-decomp}), we require
\begin{equation}
    \lim_{r\rightarrow\infty}\psi = \lim_{r\rightarrow\infty}\chi = 0,
\end{equation}
and define the following linear, invertible poloidal-toroidal projection operator:
\begin{equation}
    \mathbb{P}[\bm{u}]=\bm{U}\equiv (\psi,\,\chi)^T,
\end{equation}
where $\bm{U}$ will be referred to as the poloidal-toroidal vector henceforth. The equations of motion for the disturbance can then be expressed in terms of the disturbance vector $\bm{U}$:
\begin{equation}
    \frac{\partial \bm{U}}{\partial t} = \mathbb{N}[\bm{V}_0,\bm{U}] + \frac{1}{2}\mathbb{N}[\bm{U},\bm{U}],
    \label{eqn:pt-euler}
\end{equation}
where
\begin{equation}
\begin{gathered}
    \bm{V}_0 \equiv \mathbb{P}[\vq], \\
    \mathbb{N}[\bm{U}_1,\bm{U}_2] \equiv \mathbb{P}[\bm{u}_1\cross(\curl\bm{u}_2)-(\curl{\bm{u}_1})\cross\bm{u}_2],
\end{gathered}
\end{equation}
for arbitrary $\bm{U}_{j} = \mathbb{P}[\bm{u}_{j}]$ ($j=1,\,2$). As the poloidal-toroidal projection of any conservative vector field (i.e., gradient of a scalar potential) is identically zero, (\ref{eqn:pt-euler}) eliminates the pressure term entirely. Accordingly, the disturbance dynamics are cast in a compact operator form in terms of the poloidal-toroidal vector, while ensuring the velocity field remains solenoidal by construction.

The objective of this study is to investigate the weakly nonlinear stability of the columnar vortex (\ref{eqn:col-vortex}) with a focus on the triadic interactions between the linear wave modes. Our analysis is based on a multi-scale expansion of the disturbance equations (\ref{eqn:pt-euler}), which allows us to derive amplitude equations governing the evolution of the wave modes and to identify the conditions under which these interactions lead to instability.

\section{Multiple-scale perturbation analysis}\label{sec:multi-scale}
This section presents the general framework of multi-scale analysis for the system described by (\ref{eqn:pt-euler}). Its implications for triadic resonance in columnar vortices will be discussed in the next section.

The analysis commences by constructing an asymptotic expansion of the disturbance vector $\bm{U}$. For a small disturbance, since the nonlinearity in (\ref{eqn:pt-euler}) arises solely from the quadratic term $\mathbb{N}[\bm{U},\bm{U}]/2$, it can be expanded in powers of a scale parameter, $\epsilon \ll 1$:
\begin{equation}
    \bm{U} = \sum_{n=1}^{\infty}\epsilon^n
    \bm{U}^{(n)}(r,\phi,z,t) ,
    \label{eqn:asymptotic-expansion}
\end{equation}
where nonlinear interactions manifest at $\order{\epsilon^2}$ or higher, and the leading-order component $\bm{U}^{(1)}$ resembles linear waves. Because the base flow $\vq$ has only radial variation, the linear solutions of (\ref{eqn:pt-euler}), denoted $\bm{R}$, can be expressed as Fourier series along the azimuthal and axial directions:
\begin{equation}
    \bm{R}(r,\phi,z,t) = \Tilde{\bm{R}}(r)e^{\iu(m\phi+kz)+\sigma t},
    \label{eqn:kwave}
\end{equation}
where $m$ and $k$ are the azimuthal and axial wavenumbers, and $\sigma$ and $\Tilde{\bm{R}}$ are the corresponding complex wave frequencies and wave functions respectively. To distinguish between the linear evolution and the inherent oscillatory nature of each linear wave, the complex wave frequency is separated into the real growth rate $\lambda$ and the oscillation frequency $\omega$:
\begin{equation}
    \sigma \equiv \lambda + \iu \omega.
\end{equation}
Consider a ``slow'' timescale that is much longer than the oscillation periods of the waves: 
\begin{equation}
    \tau \equiv \epsilon t,
\end{equation}
$\bm{U}^{(1)}$ can be constructed as a superposition of linear wave modes:
\begin{equation}
    \bm{U}^{(1)} \equiv \sum_{j} A_j(\tau) \Tilde{\bm{R}}_j(r)\mathrm{exp} \big[\iu(m_j\phi+k_j z+\omega_j t) \big] + \mathrm{c.c.},
    \label{eqn:multiple-scale}
\end{equation}
where the complex conjugate (of all precedent terms), denoted as $\mathrm{c.c.}$, is added to maintain a physical disturbance. 

In this multi-scale formulation, the modulations of the wave amplitudes $A_j$ not only reflect the linear growth or decay ($\lambda$) of individual modes but also account for their nonlinear interactions at $\order{\epsilon^2}$. In the following analysis, the amplitude equations governing the evolution of $A_j$'s will be derived, and their analytical solutions will be discussed.

\subsection{The linear waves}\label{sec:linear-problem}
The linear wave modes are determined by the leading-order problem. Substituting (\ref{eqn:kwave}) into (\ref{eqn:pt-euler}) and retaining only $\order{\epsilon}$-terms yields an eigenvalue problem for the linear operator $\mathbb{M}_{mk}$, defined as:
\begin{equation}
    \sigma \Tilde{\bm{R}} = \mathbb{N}_{mk}[\bm{V}_0,\Tilde{\bm{R}}] \equiv \mathbb{M}_{mk}\Tilde{\bm{R}},
    \label{eqn:evp}
\end{equation}
where $\sigma$ and $\Tilde{\bm{R}}$ are the corresponding eigenvalue and eigenvector, and the subscripts denote the wavenumbers associated with the linear operator. 
Similarly, the linear matrix operator $\mathbb{M}_{mk}$ admits a left-hand eigenvector $\Tilde{\bm{L}}^{mk}_j$ that satisfies 
\begin{equation}
    \sigma^{mk}_j (\Tilde{\bm{L}}^{mk}_j)^\mathrm{H} = (\Tilde{\bm{L}}^{mk}_j)^\mathrm{H}\mathbb{M}_{mk},
\end{equation}
which can be normalised to attain biorthogonality with the right-hand eigenvector:
\begin{equation}
    \left \langle \Tilde{\bm{L}}^{mk}_j \middle| \Tilde{\bm{R}}^{mk}_l \right \rangle = \delta_{jl},
    \label{eqn:evp-ortho}
\end{equation}
where the inner product is defined as $\left \langle \bm{X} \middle| \bm{Y} \right \rangle \equiv \int_0^\infty \left( \bm{X} \cdot \bm{Y} \right) rdr$.

It is useful to highlight an important invariance property of the eigenvalue problem (\ref{eqn:evp}). In particular, while the physical stability properties of the vortex are intrinsic, the wave frequency $\omega$ (the imaginary part of the eigenvalue $\sigma$) is frame-dependent. In a reference frame translating with constant velocity $\bar{V}$ along the $z$-axis, the transformed linear eigenvalue problem reads:
\begin{equation}
\begin{aligned}
    \sigma' \Tilde{\bm{R}}' &= \mathbb{M}_{mk} \Tilde{\bm{R}'} + \mathbb{P}_{mk}\big[-\bar{V}\uvec_z\cross(\curl{\mathbb{P}^{-1}[\bm{R}']})\big] \\
    &= \mathbb{M}_{mk} \Tilde{\bm{R}'} +\iu k\bar{V} \Tilde{\bm{R}'},
\end{aligned}
\end{equation}
where $\sigma'$ and $\Tilde{\bm{R}}'$ are the eigenvalue and eigenvector in the moving frame.
This implies a Doppler shift in the wave frequency while the eigenvector structure remains invariant:
\begin{equation}
    \omega' = \omega + k\bar{V}, \quad \Tilde{\bm{R}'} = \Tilde{\bm{R}}.
    \label{eqn:translation2}
\end{equation}
An analogous invariance holds for a frame rotating with constant angular velocity about the vortex centreline (see Appendix~\ref{app:rotation}). 

While the spectral properties of (\ref{eqn:evp}) depend on the specific base flow profile, the time-reversibility of the underlying Euler equations implies that any eigenvalue $\sigma$ with a non-zero real part must have a conjugate counterpart with the opposite sign of the real part. Therefore, for an isolated vortex that models aircraft wake vortices or similar long-lived structures, all eigenvalues shall be close to purely imaginary, corresponding to neutral or near-neutral stability. If we restrict our attention to strictly neutral modes with $\lambda = 0$, it is evident that the transition to instability has a specific spectral requirement: a pair of neutral eigenvalues must move towards each other and collide along the imaginary axis, creating a geometric \textit{degeneracy} before they can be ``pulled'' off the axis into the unstable and decaying half-planes respectively. This mechanism, illustrated conceptually in figure~\ref{fig:concept}, acts as a necessary pathway towards instability. In the framework of weakly nonlinear theory, a simplest and most fundamental realisation of this mechanism is \textit{triadic resonance}, which provides the nonlinear coupling required to unfold the degeneracy and potentially trigger instabilities.

\begin{figure} 
    \centering 
    \includegraphics[width=.4\textwidth]{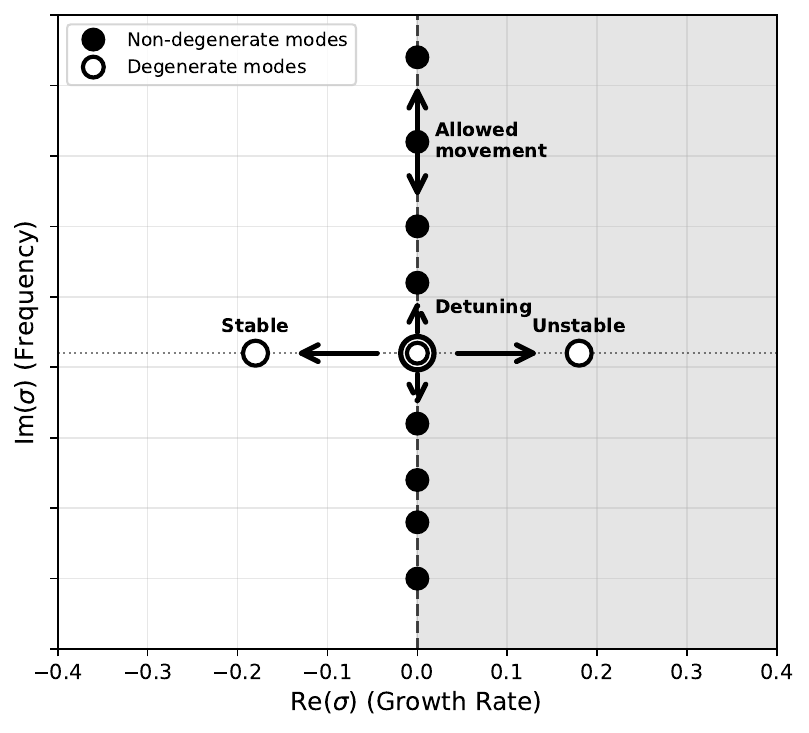}
    \caption{Stability mechanism via mode degeneracy in the inviscid limit. In the absence of degeneracy, discrete neutral modes (represented by the black circles) are topologically constrained to the imaginary axis (neutral stability) and are only permitted to move vertically along it as parameters vary. When a pair of modes becomes degenerate or nearly degenerate (represented by the white concentric circles), a resonant interaction can unfold in two distinct ways: either via detuning, where the modes separate along the imaginary axis, or via instability, where the modes are pulled symmetrically off the axis into the complex plane, resulting into a conjugate pair of growing and decaying modes.} 
    \label{fig:concept} 
\end{figure}

\subsection{Three-wave resonance}
When a disturbance field comprises a superposition of multiple linear eigenmodes, any pair of waves, denoted by the indices $p$ and $q$, interacts via the quadratic nonlinearity $\mathbb{N}[\bm{U},\bm{U}]$ of (\ref{eqn:pt-euler}) and generates a reaction term characterised by the sum of their phases:
\begin{equation*}
    \exp\big[ \iu\big((m_p+m_q)\phi+(k_p+k_q) z+(\omega_p+\omega_q)t\big)\big].
\end{equation*}
If this nonlinear product effectively projects onto a third eigenmode, a resonant triad is established. This resonant coupling facilitates efficient energy exchange between the modes, potentially destabilising a flow that is linearly stable to infinitesimal disturbance. Without loss of generality, we designate the triad members as $j = 0, 1$, and $2$, satisfying the resonant conditions described below.

The kinematic conditions required for sustained resonant interaction are:
\begin{equation}
    m_0 + m_1 = m_2, \quad k_0 + k_1 = k_2, \quad \omega_2 - (\omega_0 + \omega_1) = \epsilon\Delta\omega,
    \label{eqn:resonance}
\end{equation}
where the frequency detuning is assumed to be of $\order{\epsilon}$ or higher to prevent the interaction from averaging out over the timescales of amplitude evolution. Although the linear growth rates $\lambda_j$ do not explicitly appear in (\ref{eqn:resonance}), modes with significant damping ($|\lambda_j| \gg \epsilon$) effectively decouple from the triad. Meanwhile, modes that are already linearly unstable are of lesser interest in this weakly nonlinear context. In the rest of this section, we shall assume that all triad members are effectively neutral in the linear regime, with $|\lambda_j| \lesssim \mathcal{O}(\epsilon)$.

Equation (\ref{eqn:resonance}) can be understood as a degeneracy pulling mechanism. Specifically, consider a Galilean transformation to a reference frame moving at the same axial phase velocity as the primary wave (mode 0), $\bar{V} = -\omega_0/k_0$. In this frame, the primary wave appears as a stationary deformation ($\omega'_0 = 0$). The frequency matching requirement in (\ref{eqn:resonance}) reduces to \textit{degeneracy} or \textit{nearly-degeneracy} of the other two members:
\begin{equation}
    \omega'_2 \approx \omega'_1.
    \label{eqn:degeneracy}
\end{equation}
Physically, the deformation induced by mode $0$ breaks the symmetry of the background vortex and couples the previously orthogonal modes $1$ and $2$. This coupling allows the degenerate pair to interact, potentially causing their eigenvalues to pull off the neutral axis and leading to instability.

\subsubsection{General 3-wave amplitude equations}
To derive the equations governing the temporal evolution of the first-order wave amplitudes, we first approximate the second-order disturbance using the linear modes:
\begin{equation}
    \bm{U}^{(2)} = \sum_{m,k}\sum_{j} B^{mk}_j(\tau)\Tilde{\bm{R}}^{mk}_{j}\mathrm{exp}\big[ \iu(m\phi+k z+\omega^{mk}_j t)\big]  + \mathrm{c.c.},
\end{equation}
where slow modulations of $B^{mk}_j$'s are introduced following the method of multiple scales. 

Substituting the expansions for $\bm{U}^{(1)}$ and $\bm{U}^{(2)}$ into (\ref{eqn:pt-euler}) and imposing the solvability condition at $\mathcal{O}(\epsilon^2)$ yields the evolution equations for the triad amplitudes. Assuming $(m_j, k_j) \neq (0,0)$ to exclude self-interaction cases and retaining terms up to $\mathcal{O}(\epsilon^2)$, we obtain:
\begin{equation}
\begin{aligned}
    \dv{}{\tau}A_0 \Tilde{\bm{R}}_0 - \frac{\lambda_0}{\epsilon}A_0 \Tilde{\bm{R}}_0 - \sum_j \lambda^{m_0k_0}_jB^{m_0k_0}_j \Tilde{\bm{R}}^{m_0k_0}_j &= \mathbb{N}_{m_0k_0}[\Tilde{\bm{R}}_1^*,\Tilde{\bm{R}}_2]A_1^*A_2e^{\iu\Delta\omega \tau}, \\
    \dv{}{\tau}A_1 \Tilde{\bm{R}}_1 - \frac{\lambda_1}{\epsilon}A_1 \Tilde{\bm{R}}_1 - \sum_j \lambda^{m_1k_1}_jB^{m_1k_1}_j \Tilde{\bm{R}}^{m_1k_1}_j &= \mathbb{N}_{m_1k_1}[\Tilde{\bm{R}}_0^*,\Tilde{\bm{R}}_2]A_0^*A_2e^{\iu\Delta\omega \tau}, \\
    \dv{}{\tau}A_2 \Tilde{\bm{R}}_2 - \frac{\lambda_2}{\epsilon}A_2 \Tilde{\bm{R}}_2 - \sum_j \underset{\sim\order{\epsilon}\mbox{ or higher}}{\underbrace{{\lambda^{m_2k_2}_jB^{m_2k_2}_j \Tilde{\bm{R}}^{m_2k_2}_j}}} &= \mathbb{N}_{m_2k_2}[\Tilde{\bm{R}}_0,\Tilde{\bm{R}}_1]A_0A_1e^{-\iu\Delta\omega \tau}.
\end{aligned}
\label{eqn:3-wave-middle}
\end{equation}

Left-multiplying (\ref{eqn:3-wave-middle}) with the adjoint eigenvectors $\Tilde{\bm{L}}_j$ and applying the biorthogonality condition (\ref{eqn:evp-ortho}), we arrive at the canonical amplitude equations for an arbitrary resonant triad:
\begin{equation}
\begin{aligned}
    & \dv{}{\tau}A_0 -\frac{\lambda_0}{\epsilon}A_0 =  J_0 A_1^*A_2e^{\iu\Delta\omega \tau},\\
    & \dv{}{\tau}A_1 - \frac{\lambda_1}{\epsilon}A_1 = J_1 A_0^*A_2e^{\iu\Delta\omega \tau},\\
    & \dv{}{\tau}A_2 - \frac{\lambda_2}{\epsilon}A_2 = J_2 A_0A_1e^{-\iu\Delta\omega \tau},
\end{aligned}
\label{eqn:3-wave-general}
\end{equation}
with the nonlinear interaction coefficients defined as follows:
\begin{equation}        
\begin{aligned}
    J_0 \equiv \left \langle \Tilde{\bm{L}}_0 \middle| \mathbb{N}_{m_0k_0}[\Tilde{\bm{R}}_1^*,\Tilde{\bm{R}}_2] \right \rangle, \\ 
    J_1 \equiv \left \langle \Tilde{\bm{L}}_1 \middle| \mathbb{N}_{m_1k_1}[\Tilde{\bm{R}}_0^*,\Tilde{\bm{R}}_2] \right \rangle, \\ 
    J_2 \equiv \left \langle \Tilde{\bm{L}}_2 \middle| \mathbb{N}_{m_2k_2}[\Tilde{\bm{R}}_0,\Tilde{\bm{R}}_1] \right \rangle.
\end{aligned}
\label{eqn:interaction-coeff}
\end{equation}


When the wave functions of a resonant triad are known, its interaction coefficients can be calculated from (\ref{eqn:interaction-coeff}), and the numerical solutions of (\ref{eqn:3-wave-general}) can be easily computed. It is important to recall from the initial multi-scale formulation that the $\lambda_j$ terms explicitly encapsulate the linear growth or decay rates of the individual modes. As long as these linear rates are numerically small, scaling at $\order{\epsilon}$ or higher order, they compete directly with the weakly nonlinear interactions. In such regimes, the general three-wave amplitude equations successfully accommodate linearly damped or growing modes, yielding a satisfactory approximation of the dynamics at sufficiently small but finite wave amplitudes \citep{craik_1986}. Analytical approaches, on the other hand, have thus far yielded only a limited number of qualitative findings for the general three-wave resonance, with no universally acknowledged general analytic solutions. Therefore, in the following analysis, we will focus on the case of conservative interactions, where analytical insights into (\ref{eqn:3-wave-general}) can be obtained.

\subsubsection{Conservative interactions and explosive resonance}\label{sec:conservative-triad}
Consider the scenario where a resonant triad satisfies the following hypothesis:
\begin{hyp}\label{hyp:conservative-condition}
    All three interaction coefficients of the resonant triad are purely imaginary: $J_0, J_1, J_2 \in \iu\mathbb{R}$.
\end{hyp}
In the ideal limit of exact resonance ($\Delta\omega=0$) between neutral modes ($\lambda = 0$), (\ref{eqn:3-wave-general}) simplify to
\begin{equation} 
    \dv{}{\tau}A_0 = J_0 A_1^*A_2,\quad
    \dv{}{\tau}A_1 = J_1 A_0^*A_2,\quad
    \dv{}{\tau}A_2 = J_2 A_0A_1 .
    \label{eqn:3-wave-conservative}
\end{equation}
Defining the wave intensity of each wave mode as
\begin{equation}
    E_j \equiv |A_j|^2 \in \mathbb{R}^+,
\end{equation}
and utilising the property $J_j^* = -J_j$ for purely imaginary coefficients, we recover the Manley--Rowe relations:
\begin{equation}
    \frac{d_\tau E_0}{J_0} = \frac{d_\tau E_1}{J_1} = -\frac{d_\tau E_2}{J_2} = \Gamma
    \label{eqn:manley-rowe}
\end{equation}
where $\Gamma(\tau) \equiv A_0^*A_1^*A_2 - A_0A_1A_2^*$. 
These relations imply the existence of conserved quantities (constants of motion) of the form
\begin{equation}
    C_0^+\equiv\frac{E_0}{J_0}+\frac{E_2}{J_2},\quad
    C_1^+\equiv\frac{E_1}{J_1}+\frac{E_2}{J_2},\quad
    C^-\equiv\frac{E_0}{J_0}-\frac{E_1}{J_1}.
\label{eqn:manley-rowe-constant}
\end{equation}
Therefore, resonant triads satisfying Hypothesis~\ref{hyp:conservative-condition} form a conservative dynamical system, where the exchange of wave intensities is strictly confined to a compact manifold defined by the constants of motion. It is worth noting that there exists a more general condition for conservative resonance (for $p,q = 0, 1, 2$):
\begin{equation*}
    \arg[J_p]-\arg[J_q] = n\pi, \quad n\in\mathbb{Z},
\end{equation*}
and an equivalent set of constants of motion can be derived. However, as will be elaborated in \S\ref{sec:selection-rules}, the triadic interactions of smooth neutral modes in columnar vortices inherently satisfy Hypothesis~\ref{hyp:conservative-condition}, which justifies our focus on this particular scenario.

\begin{figure}
    \centering
    \includegraphics[width=\textwidth]{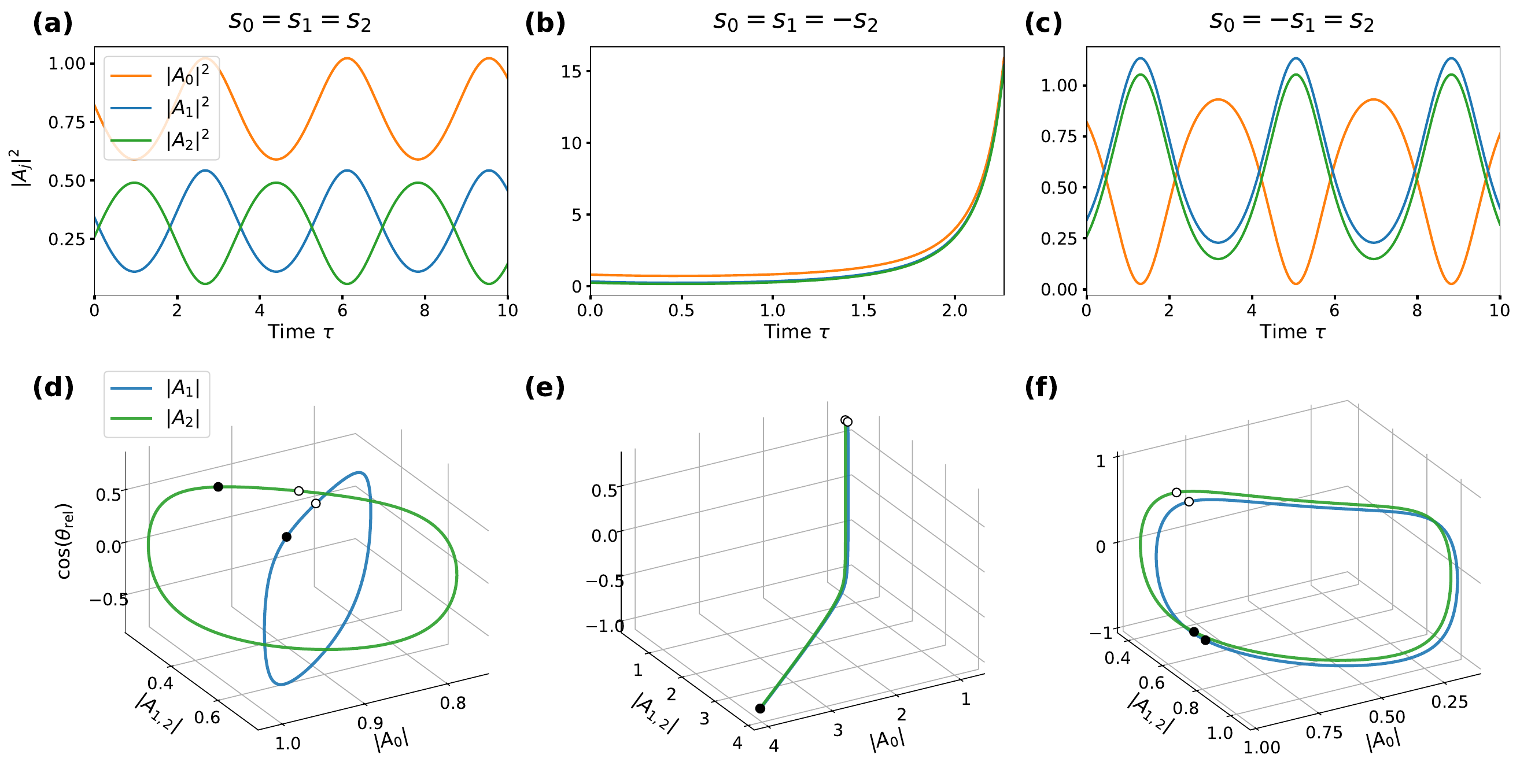}
    \caption{Example numerical solutions of the normalised conservative three-wave amplitude equations (\ref{eqn:3-wave-conservative-universal}) for distinct interaction coefficient sign signatures ($s_0$, $s_1$, $s_2$). (a--c) Temporal evolution of the normalised wave intensities $|\hat{A}_j|^2$. (d--f) Corresponding 3D phase space trajectories plotting $|\hat{A}_{1,2}|$ against $|\hat{A}_0|$ and the relative phase factor $\cos{(\phi_2-\phi_1-\phi_0)}$. White and black markers indicate the initial and final states of the evolution, respectively. Left and Right Columns: Bounded resonance regimes, where the system is strictly confined to closed periodic orbits on a compact manifold. Middle Column: The explosive resonance regime ($s_0=s_1=-s_2$), where the relative phase locks to a constant value while permitting the system to follow an open trajectory with simultaneous amplitude divergence. See also \citet{craik_1986}.}
    \label{fig:evolution}
\end{figure}

The analytical solutions of the amplitude equations (\ref{eqn:3-wave-conservative}) are expressible in terms of the elliptic functions \citep[see][]{weiland_1977} and consist of two distinct branches: bounded periodic oscillations or simultaneous explosive growth. The selection between these branches is determined solely by the relative signs of the interaction coefficients. To demonstrate this, we normalise (\ref{eqn:3-wave-conservative}) into a universal form:
\begin{equation}
    \dv{}{\tau}\hat{A}_0 = -s_0 \hat{A}_1^*\hat{A}_2,\quad
    \dv{}{\tau}\hat{A}_1 = -s_1 \hat{A}_0^*\hat{A}_2,\quad
    \dv{}{\tau}\hat{A}_2 = s_2 \hat{A}_0\hat{A}_1 ,
    \label{eqn:3-wave-conservative-universal}
\end{equation}
where the normalised amplitudes and sign indicators are defined as:
\begin{equation*}
    \left\{
    \begin{matrix}
    \hat{A}_j = K\frac{A_j}{\sqrt{|J_j|}}, \\
    s_j = -\iu J_j/|J_j|, \\
    \end{matrix}
    \right. 
\end{equation*}
with $K \equiv \iu\sqrt{|J_0||J_1||J_2|}$ and $s_j \in \{+1, -1\}$. The numerical solutions for (\ref{eqn:3-wave-conservative-universal}) can be easily obtained, and a series of example solutions for different combinations of $s_j$'s are given in figure~\ref{fig:evolution}. Following the derivation of the Manley--Rowe relations, we obtain the normalised constants of motion:
\begin{equation}
    \hat{C}_0^+\equiv s_0\hat{E}_0+s_2\hat{E}_2,\quad
    \hat{C}_1^+\equiv s_1\hat{E}_1+s_2\hat{E}_2,\quad       
    \hat{C}^-\equiv s_0\hat{E}_0-s_1\hat{E}_1,
    \label{eqn:manley-rowe-constant-normalised}
\end{equation}
where $\hat{E}_j = |\hat{A}_j|^2$ represents the normalised positive-definite wave intensity. The dynamics represented by (\ref{eqn:manley-rowe-constant-normalised}) have two topological regimes: 
\begin{enumerate}
    \item The first regime, \textit{explosive resonance}, occurs strictly when the coefficient of the sum mode (mode $2$, defined by $\omega_2 = \omega_0 + \omega_1$) opposes those of both constituent modes:
    \begin{equation}
        s_0 = s_1 = -s_2 .
        \label{eqn:sign-of-j}
    \end{equation}
    In this specific topology, the invariants $\hat{C}^+$ become differences of wave intensities (e.g., $\hat{E}_2 - \hat{E}_0 = \text{const}$), so the conservation laws place no upper limit on the system, allowing all three wave amplitudes to diverge simultaneously as $(t_c - t)^{-1}$ towards a finite-time singularity $t_c$, as shown in figure~\ref{fig:evolution}b.
    \item The second regime, \textit{bounded resonance}, encompasses all other sign combinations. Regardless of whether the signs are uniform ($s_0=s_1=s_2$) or mixed (e.g., $s_2 = s_0 = -s_1$), at least one of the invariants in (\ref{eqn:manley-rowe-constant-normalised}) serves as a sum of positive intensities, defining a bounding sphere or ellipsoid in phase space and forcing the system to undergo periodic amplitude oscillations without divergence, as illustrated in figures~\ref{fig:evolution}a and \ref{fig:evolution}c.
\end{enumerate}  
As indicated by figure~\ref{fig:evolution}d-f, the relative phase factor $\cos{(\phi_2-\phi_1-\phi_0)}$ remains constant for the explosive regime, indicating a phase locking mechanism that sustains the growth; in contrast, the bounded resonance cases exhibit periodic variations in the relative phase, which prevents the system from diverging and confines it to closed orbits in phase space.

Lastly, it is important to address the apparent paradox of explosive growth in a conservative system. The resolution lies in the distinction between the wave intensity $E_j$ tracked by the amplitude equations and the thermodynamic energy of the perturbation. In the presence of a background flow, a mode can extract energy from the mean flow, which further sustains the growth of wave intensities of the triad members while preserving the global energy balance of the fluid system. However, as noted by \citet{weiland_1977} and \citet{craik_1986}, this growth eventually outpaces the weakly nonlinear regime. At this threshold, higher-order effects — specifically third-order terms — may introduce frequency detuning or saturation that suppresses the singularity. Nevertheless, explosive triadic resonance remains a critical mechanism because it requires no sustained external forcing to drive the growth; instead, it provides a fast pathway for infinitesimal disturbances to grow rapidly, even if higher-order effects eventually saturate the interaction as the triad members reach large amplitudes.


\subsubsection{Parametric instability via external forcing}
In addition to the explosive instability, the conservative three-wave resonance can also manifest as a parametric instability when one mode is externally forced and serves as a pump to drive the growth of the other two modes. Without loss of generality, we assume $A_0$ as a frozen mode with constant amplitude. Under this assumption, the general amplitude evolution equations (\ref{eqn:3-wave-conservative}) reduce to a coupled linear system for the two free modes $(j = 1, 2)$:
\begin{equation}
    \dv{}{\tau}A_0 = 0,\quad 
    \dv{}{\tau}A_1 = J_1 A_0^*A_2,\quad
    \dv{}{\tau}A_2 = J_2 A_0A_1 .
    \label{eqn:3-wave-parametric}
\end{equation}
Differentiating with respect to $\tau$ decouples the system, yielding a second-order ordinary differential equation: 
\begin{equation}
    \dv[2]{}{\tau}A_j = |A_0|^2J_1J_2 A_j,
\end{equation}
which admit solutions of the form:
\begin{equation}
    A_j(\tau) = A_j^0\, e^{\pm\sqrt{J_1J_2}|A_0|\cdot\tau}.
    \label{eqn:3-wave-parametric-sol}
\end{equation}
Substituting (\ref{eqn:3-wave-parametric-sol}) back into (\ref{eqn:3-wave-parametric}) reveals the constant coupling between the free modes:
\begin{equation} 
    \frac{A_2}{A_1} = \pm\sqrt{\frac{J_2}{J_1}}\arg[A_0].
\end{equation}
Recall that $\tau$ is the slow timescale of $\order{\epsilon}$, so the total complex frequency (eigenvalue) for the coupled system is the sum of the original linear complex frequency $\sigma_j$ and the parametric correction:
\begin{equation}
    \sigma_{j,\text{par}} = \sigma_j \pm \epsilon \sqrt{J_1J_2}|A_0|,
    \label{eqn:corrected-parametric-sigma}
\end{equation}
which is either periodic when $\sgn[J_1]\cdot\sgn[J_2] > 0$ ($J_1$ and $J_2$ are both purely imaginary under Hypothesis~\ref{hyp:conservative-condition}), or exponential and implying instability when $\sgn[J_1]\cdot\sgn[J_2] < 0$. In the latter case, the conjugate corrections ($\pm$) exactly correspond to the eigenvalues being ``pulled'' symmetrically off the neutral axis while strictly preserving the time-reversibility of the underlying Hamiltonian system. 

\begin{figure}
    \centering
    \includegraphics[width=\textwidth]{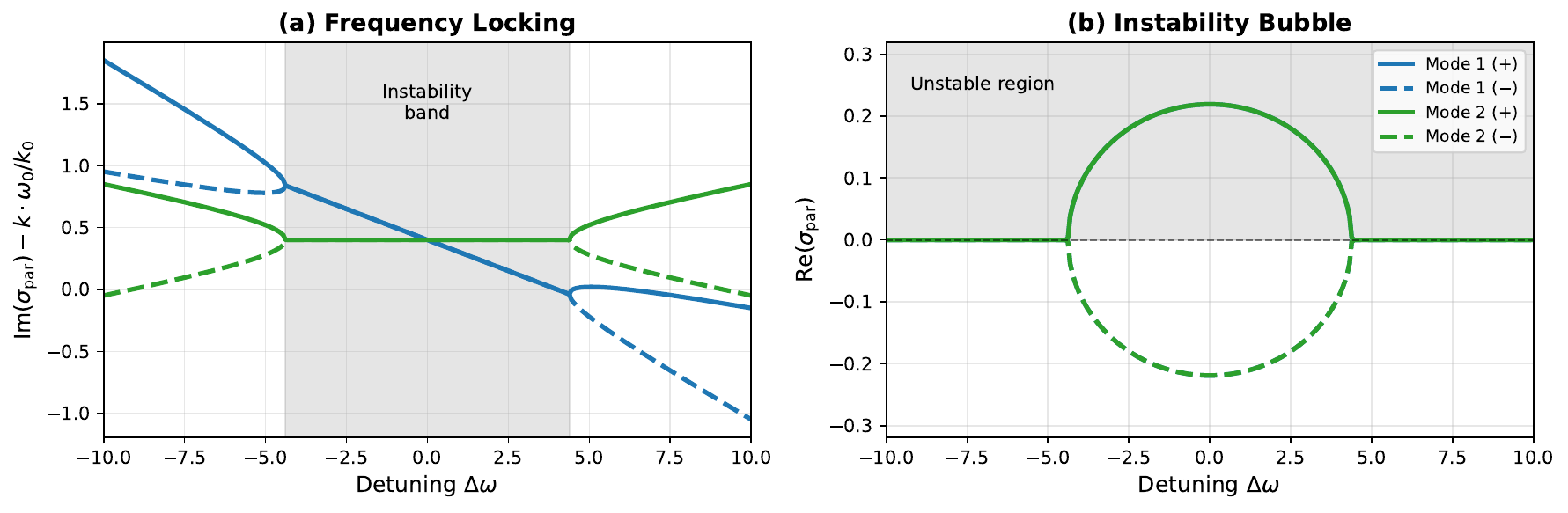}
    \caption{Effect of detuning $\Delta\omega$ on parametric instability as given in (\ref{eqn:sigma-detuned}). (a) Doppler-shifted eigenfrequencies $\mathrm{Im}(\sigma_{\text{par},j}) - k_j\omega_0/k_0$ for modes $j=1,2$. The solid and dashed lines represent the upper and lower branches of the four coupled eigenvalues, respectively. Inside the instability band (shaded region), the upper and lower branches of the same mode locks to the same wave frequency, and the eigenvalues merge toward exact resonance where the modes synchronise to a common Doppler-shifted frequency. (b) Real growth rates $\mathrm{Re}(\sigma)$. In the region where frequencies are locked, the growth rates bifurcate to form a symmetric instability bubble (shaded region indicating positive growth). The two wave modes share identical growth rates, leading to the appearance of a single curve for the unstable branch. Outside this region, the frequency mismatch becomes significant, and the coupled modes are no longer in near resonance.}
    \label{fig:detuning}
\end{figure}

In realistic flows, the resonance condition may not be perfectly satisfied. Our general three-wave resonance formulation allows a detuning of $\epsilon\Delta\omega = \omega_2 - (\omega_1 + \omega_0)$. Following the same decoupling procedure as the perfect resonance case, the reduced second-order differential equation for the free modes now includes a first-order term representing the phase drift. For the free mode with amplitude $A_2$, this is expressed as
\begin{equation}
    \dv[2]{}{\tau}A_2 + \iu \Delta\omega \dv{}{\tau}A_2 - |A_0|^2 J_1 J_2 A_2 = 0.
\end{equation}
Assuming a solution consistent with the slow-time expansion, the corrected complex wave frequencies are now derived as:
\begin{equation}
    \sigma_{j,\text{par}} = \sigma_j - \iu \frac{\epsilon \Delta\omega}{2} \pm \epsilon \sqrt{|A_0|^2 J_1 J_2 - \left(\frac{\Delta\omega}{2}\right)^2}.
    \label{eqn:sigma-detuned}
\end{equation}
Figure~\ref{fig:detuning} illustrates the influence of detuning on the parametric instability based on (\ref{eqn:sigma-detuned}). The presence of detuning introduces a threshold for instability, leading to a band of detuning values. Specifically, when the detuning is small enough such that $0 \leq (\Delta\omega/2)^2 < |A_0|^2 J_1 J_2$, the radical in (\ref{eqn:sigma-detuned}) becomes real, resulting in a pair of complex conjugate eigenvalues with peak growth rates reached at perfect resonance ($\Delta\omega=0$). Conversely, if the detuning is too large, the radical remains imaginary, and the instability does not exist. In \S\ref{sec:tuning}, a numerical method based on the non-degenerate perturbation theory will be introduced to locate and tune triads for minimal detuning, thereby maximising the growth rates predicted by (\ref{eqn:sigma-detuned}).

Finally, we note that the frozen mode assumption is also valid for unforced triads but only in the regime where $|A_0|\gg |A_1|,|A_2|$. As the perturbation amplitudes eventually become comparable, the depletion effects become significant. In the absence of external forcing to replenish $A_0$, the system inevitably transitions to the general three-wave conservative interactions confined by the Manley--Rowe relations (\ref{eqn:manley-rowe-constant}).

\section{Implications for columnar vortices}\label{sec:selection-rules}
In this section, we apply the general framework of three-wave resonance to the specific context of columnar vortices. We begin by classifying the linear stability properties of columnar vortices, with a particular focus on the characteristics of neutral modes that are relevant for triadic interactions. Building on this classification, we derive a set of selection rules that govern the nonlinear coupling between these modes, which in turn determine the possible resonance scenarios and their associated dynamical behaviours.

\subsection{Linear stability properties}
The linear stability of columnar vortices has been the subject of extensive research \citep[e.g.,][]{mayer_powell_1992, jacquin_2006, lee_marcus}. Rather than reiterating the full breadth of these studies, our main objective here is to provide a clear classification of the neutral modes ($\sigma = \iu\omega$), which represent the oscillatory degrees of freedom in the inviscid limit available for the expansion in (\ref{eqn:multiple-scale}) and subsequent triadic resonance coupling. 

To facilitate this classification, we consider the reduction of (\ref{eqn:evp}) in terms of the primitive variables, which leads to the Howard--Gupta equation \citep{howard_gupta_1962} expressed here in terms of the radial velocity component $\Tilde{u}_r(r)$:
\begin{equation}
    \dv{}{r}\left(S\dv{(r \Tilde{u}_r)}{r}\right) - \left(\frac{a(r)}{\Phi}+\frac{b(r)}{\Phi^2}+r\right)\frac{\Tilde{u}_r}{r} = 0,
    \label{eqn:howard-grupta}
\end{equation}
where the auxiliary functions are defined as:
\begin{equation}
    \begin{aligned}
        \Phi(r) &= -\iu\sigma + m\Omega(r) + k V_z(r), \\
        S(r) &= \frac{r^2}{k^2r^2+m^2}, \\
        a(r) &= r \dv{}{r}\left[\frac{S}{r}\left(\dv{\Phi}{r}+\frac{2m}{r}\Omega\right)\right], \quad
        b(r) = 2kmS\Omega\left[\frac{1}{r}\dv{V_z}{r} - \frac{k}{m}\zeta\right].
    \end{aligned}
\end{equation}

For the neutral modes of interest, the coefficients in (\ref{eqn:howard-grupta}) are purely real. Therefore, when in the absence of any singularities, the eigenfunctions $\Tilde{u}_r(r)$ possess a global phase coherence. Meanwhile, the condition of mass conservation,
\begin{equation}
    \frac{1}{r}\dv{}{r} {(r\Tilde{u}_r)} + \frac{\iu m}{r}\Tilde{u}_\phi + \iu k \Tilde{u}_z = 0,
    \label{eqn:continuity_modal}
\end{equation}
imposes a quadrature relation ($90^\circ$ phase shift) between the phase of $\Tilde{u}_r$ and the phases of the azimuthal and axial velocity components. Without loss of generality, we apply a global phase shift to all neutral, non-singular solutions in our subsequent analysis such that the radial velocity is purely imaginary ($\Tilde{u}_r \in \iu\mathbb{R}$), rendering the tangential components real ($\Tilde{u}_\phi, \Tilde{u}_z \in \mathbb{R}$).

The solutions to (\ref{eqn:howard-grupta}) can carry critical layer singularities, which occur at locations where the Doppler-shifted frequency $\Phi(r)$ vanishes. The nature of the eigenfunctions and their associated dispersion relations are fundamentally influenced by whether such critical layers exist, which leads to two general wave categories:
\begin{enumerate}
    \item \textit{Regular modes} (or Kelvin waves), which exist when $\Phi(r) \neq 0$ for all $r$. In this case, the Howard--Gupta equation remains regular, yielding eigenfunctions that are analytic and possess the phase coherence described above. Their dispersion relations, $D(\omega, m, k) = 0$, are discrete and well-approximated by the large-$k$ asymptotic theory \citep{dizes_2005_2}.
    
    \item \textit{Critical layer solutions}, which arise when $\Phi(r_c) = 0$ at some critical radius $r_c$ and constitute a \textit{continuous spectrum} in the dispersion plot. At the critical radius, the equation is singular, typically resulting in a logarithmic singularity and a phase jump that breaks the global phase coherence \citep{gallay_smets_2020}. Following the classical treatments of the continuous spectrum in parallel shear flows and vortex columns, we refer to these generalised solutions as the \textit{singular eigenfunctions} of the continuous spectrum \citep{case_1960, jacquin_2006, roy_subramanian_2014}.
\end{enumerate}

Notably, within the continuous spectrum, there exists a discrete set of \textit{passive} critical layer modes that maintain structural regularity despite the presence of critical layer singularity. These modes can be interpreted as the continuation of discrete dispersion branches into the continuous spectrum \citep{jacquin_2006}, whose dispersion relations can still be approximated within the large-$k$ WKBJ framework developed by \citet{dizes_2005_2}. 

\begin{figure}
    \centering
    \includegraphics[width=\linewidth]{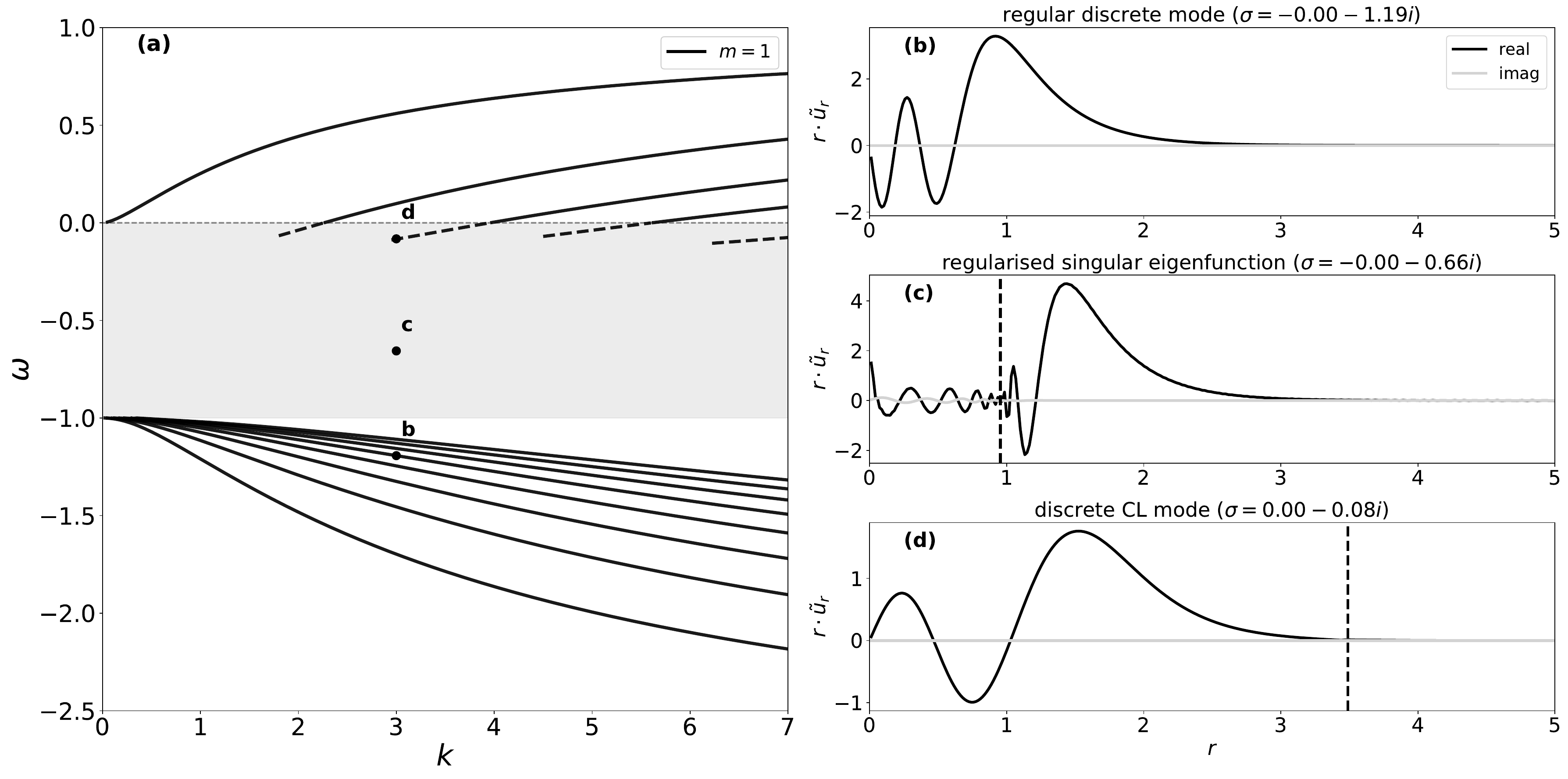}
    \caption{Characteristics of the inviscid linear modes of the Lamb--Oseen vortex obtained numerically by MLEGS \citep{MLEGS}. (a) Dispersion curves $\omega(k)$ of helical modes ($m=1$). Solid lines denote Kelvin waves, and dashed lines indicate neutral discrete modes possessing critical layers. The shaded frequency band $\omega_\text{CL} = [-m,0)$ is associated with the critical layer singularity and corresponds to the continuous spectrum. (b--d) Radial velocity component $r\tilde{u}_r(r)$ of selected eigenvectors ($m=1$, $k=3.0$): (b) A regular discrete Kelvin wave, exhibiting a purely real, smooth profile characteristic of global neutral modes; (c) A singular eigenfunction of the continuous spectrum, shown in its numerically regularised representation. The singularity at the critical radius $r_c$ (vertical dashed line) breaks analytic continuity, resulting in spectral broadening and a distinct phase jump in physical space; (d) A ``passive'' discrete critical layer mode. Unlike the singular eigenfunctions, this solution maintains high-order spectral convergence, confirming that the critical layer singularity is effectively suppressed and decoupled from the core dynamics. The eigenvalues associated with the mode profiles shown are marked as circles in the dispersion plot.}
    \label{fig:profiles}
\end{figure}

Figure~\ref{fig:profiles} illustrates the dispersion characteristics of the Lamb--Oseen vortex, highlighting regular Kelvin waves, the continuous spectrum, and the discrete critical layer modes. The numerical results presented here and in subsequent sections are obtained using a spectral collocation method based on the mapped associated Legendre functions (MLEGS) \citep{matsushima_marcus_1997, lee_marcus, MLEGS}, which is detailed in Appendix~\ref{appendix:mlegs}. Strictly speaking, the mathematical singularity at the critical layer that arises in the linear inviscid limit demands a regularisation mechanism, which is typically achieved via complex contour deformation. In our approach, the spatial discretisation of MLEGS implicitly provides the necessary regularisation. As thoroughly analysed and validated by \citet{lee_marcus}, this approach allows MLEGS to accurately capture both regular modes and regularised representations of the singular eigenfunctions directly without requiring integration in the complex plane.
The radial velocity profiles of representative modes are depicted in figure~\ref{fig:profiles}b--d. It can be seen that, for the discrete critical layer modes (figure~\ref{fig:profiles}d), the critical layer $r_c$ is located far from the region where the eigenmode is localised, and the singularity is effectively suppressed, resulting in a smooth eigenfunction in the inner region that is structurally indistinguishable from a regular mode. In contrast, the surrounding singular eigenfunctions of the continuous spectrum (figure~\ref{fig:profiles}c) display distinct phase jumps at the critical radius and lack global phase coherence.

For the purpose of our nonlinear analysis, we categorise the neutral modes based on structural analyticity rather than simply the existence of critical radius. In the subsequent analysis, we shall focus on the set of \textit{smooth neutral modes} who maintain analytic continuity and global phase coherence within their localised region. This group naturally includes all regular Kelvin waves and the discrete critical layer modes.

\subsection{Conservative interactions of smooth neutral modes}
We now show that for columnar vortices, the interaction of \textit{smooth neutral modes} (as classified above) is inherently conservative. This ensures that the dynamics are governed strictly by the Manley--Rowe relations derived in the previous section.

Recall that due to mass conservation, the radial velocity component $\Tilde{u}_r$ is in quadrature with the azimuthal and axial components. Following the phase convention established earlier, let $x_j$'s be real functions of $r$ such that:
\begin{equation}
    \begin{pmatrix}
    \Tilde{u}_r \\ \Tilde{u}_\phi \\ \Tilde{u}_z
    \end{pmatrix}
    =
    \begin{pmatrix}
    \iu x_1 \\ x_2 \\ x_3
    \end{pmatrix}.
\end{equation}
The vorticity field $(\Tilde{\omega}_r,\Tilde{\omega}_\phi,\Tilde{\omega}_z)e^{\iu(m\phi+kz)+\sigma t}$ inherits this phase structure. Specifically:
\begin{equation}
    \begin{pmatrix} \Tilde{\omega}_r \\ \Tilde{\omega}_\phi \\ \Tilde{\omega}_z \end{pmatrix} 
    = 
    \begin{pmatrix} \frac{1}{r}\partial_\phi \Tilde{u}_z - \partial_z \Tilde{u}_\phi \\ \partial_z \Tilde{u}_r - \partial_r \Tilde{u}_z \\ \frac{1}{r}\partial_r(r\Tilde{u}_\phi) - \frac{1}{r}\partial_\phi \Tilde{u}_r \end{pmatrix}
    = 
    \begin{pmatrix} \iu(\frac{m}{r}u_z - k u_\phi) \\ -k u_r - \partial_r u_z \\ \frac{1}{r}\partial_r(r u_\phi) + \frac{m}{r}u_r \end{pmatrix}
    =
    \begin{pmatrix}
    \iu x_4 \\ x_5 \\ x_6
    \end{pmatrix}.
\end{equation}
Thus, both velocity and vorticity vectors have the structure $(\text{Imag}, \text{Real}, \text{Real})^T$.

The nonlinear interaction term $\mathbb{N}$ involves the cross product $\Tilde{\bm{u}} \times \Tilde{\bm{\omega}}$. The components of this product follow the pattern:
\begin{equation}
    \begin{pmatrix}
    \iu x_1 \\ x_2 \\ x_3
    \end{pmatrix} 
    \cross
    \begin{pmatrix}
    \iu x_4 \\ x_5 \\ x_6
    \end{pmatrix} 
    =
    \begin{pmatrix}
    x_7 \\ \iu x_8 \\ \iu x_9
    \end{pmatrix}.
\end{equation}
Therefore, the nonlinear forcing term has the structure $(\text{Real}, \text{Imag}, \text{Imag})^T$. Finally, to compute the interaction coefficient $J$, we take the inner product with the adjoint $\Tilde{\bm{L}}$ in the primitive form. Since the adjoint eigenvectors for smooth neutral modes share the same phase structure as the direct modes ($\Tilde{L}_r \in \iu\mathbb{R}, \Tilde{L}_{\phi,z} \in \mathbb{R}$), the inner product yields:
\begin{equation}
    J \sim \int (\Tilde{L}_r^* N_r + \Tilde{L}_\phi^* N_\phi + \Tilde{L}_z^* N_z) r dr \sim (-\iu)(\text{Re}) + (\text{Re})(\iu) + (\text{Re})(\iu) \in \iu\mathbb{R}.
\end{equation}
This proves that the interaction coefficients are purely imaginary, as assumed in Hypothesis~\ref{hyp:conservative-condition}. Hence, resonant interactions involving smooth neutral modes in columnar vortices are strictly conservative, governed by the Manley--Rowe relations. This finding sets the stage for the subsequent analysis of whether such conservative triads can exhibit explosive instability or are restricted to bounded oscillations.

\subsection{Pseudoenergy and selection rules for conservative triads}
Given the potentially infinite space of possible resonant triads as evidenced by the continuous variation of axial wavenumbers, a brute-force verification of explosive instability through calculations of the nonlinear interaction coefficients is impractical. To address this, we follow \citet{cairns_1979}, who leveraged the concept of pseudoenergy from plasma physics and proposed that the onset of hydrodynamic instability can be determined by the relative signs of the triad members' pseudoenergy. The pseudoenergy can be defined to have the form\footnote{The negative sign in definition (\ref{eqn:pseudoenergy}) is introduced to ensure that regular Kelvin modes possess positive pseudoenergy. This accounts for the difference in temporal conventions: \citet{dizes_2008} defines the wave phase as $\exp[\iu(kz + m\phi - \omega t)]$, whereas our formulation assumes $\exp[\iu(kz + m\phi + \omega t)]$. This sign reversal in the time domain necessitates the adjustment in the energy definition to maintain physical consistency with past literature.}:
\begin{equation}
    \mathcal{E}_j \equiv -\frac{1}{4}\omega_j\pdv{D}{\omega_j}|A_j|^2,
    \label{eqn:pseudoenergy}
\end{equation}
where $D(\omega,m,k)$ is the dispersion relation of the wave mode. 

For stratified shear flows characterised by piece-wise constant or linear velocity profiles with piece-wise constant or exponential density profiles, studies have identified the existence of explosive resonant triads \citep{craik_adam_1979, Tsutahara_1984, Tsutahara_1986}. These results are consistent with the criterion proposed by \citet{cairns_1979}, where instability is observed when the highest-frequency mode possesses a pseudoenergy sign opposite to that of the other triad members. In the case of unstratified parallel shear flows with smooth profiles, \citet{becker_1993} demonstrated that critical layer modes are the sole carriers of negative pseudoenergy, identifying them as prerequisites for explosive resonance.

For vortical flows, \citet{fukumoto_2005} employed the pseudoenergy to investigate the parametric instability of the Rankine vortex. They posited that a wave pair in resonance via external strain must possess pseudoenergies of either opposite signs or both zero to trigger instability, a criterion they corroborated via direct eigenvector evaluation. Furthermore, they derived the pseudoenergy for two-dimensional waves directly from the total fluid kinetic energy, providing a partial justification for the application of (\ref{eqn:pseudoenergy}) to vortical flows. Subsequently, \citet{dizes_2008} extended this framework to vortices with smooth velocity profiles, such as the Lamb--Oseen vortex. By evaluating (\ref{eqn:pseudoenergy}) in the inviscid regime via a WKBJ approach \citep[see][]{dizes_2005_2}, they demonstrated that for unstratified fluids, modes with negative pseudoenergy are confined to specific frequency bands associated with the critical layer singularities. Hence, these modes are deemed essential components for the occurrence of the elliptical instability, a finding confirmed by subsequent numerical stability analyses \citep{dizes_2008,dizes_2017}. In the analysis below, we follow the large-$k$ WKBJ framework of \citet{dizes_2005_2} to investigate the pseudoenergy criterion and discuss the stability properties of the general three-wave resonance in columnar vortices.

We first note that, since the nonlinear interaction coefficients ($J_0$, $J_1$, and $J_2$) are solely calculated from the linear eigenvectors, they are invariant under the change of the viewing frame due to the Galilean invariance of the linear eigenvalue problem (refer to \S\ref{sec:linear-problem}). Hence, an explosive triad determined by (\ref{eqn:sign-of-j}) in the original frame remains so in the new frame, and vice versa. A direct application of this invariance is that, instead of evaluating the triad's pseudoenergy in the original frame, we can choose a specific reference frame in which a triad member appears stationary. In such a frame (denoted by primes), the wave frequencies of the now degenerate modes are exactly the same, and the pseudoenergy criterion used by \citet{fukumoto_2005} and \citet{dizes_2008} can be evaluated by the relative signs of the dispersion slopes:
\begin{equation}
    \sgn\left[J_p\cdot J_q\right] = \sgn\left[\pdv{D'}{\omega'_p}\cdot\pdv{D'}{\omega'_q}\right].
    \label{eqn:sign-of-j-dispersion}
\end{equation}

To determine the signs of the dispersion slopes, we employ the WKBJ framework, which has been utilised to approximate the linear dispersion relations and eigenvectors of columnar vortices \citep{dizes_2005_2, dizes_2008}. Because this asymptotic approach is formally derived for the large-$k$ limit, its accuracy varies across the parameter space. As demonstrated by \citet{dizes_2005_2} for continuous profiles with Gaussian vorticity distributions, the WKBJ predictions align closely with exact numerical calculations for low-order azimuthal modes, but exhibit deviations as the azimuthal wavenumber $m$ increases, particularly at very small $k$. Because our pseudoenergy analysis fundamentally relies on evaluating the relative signs of the dispersion slopes, any quantitative deviations in the predicted topology of the dispersion branches will inherently impact the accuracy of the sign-based criterion in these specific regimes. Therefore, while the WKBJ formulation provides a highly tractable analytical framework for investigating the general properties of triadic resonance, the subsequent stability conclusions are bounded by this asymptotic context. 

Specifically, in the large-$k$ limit, the $k^2$-proportional terms in the Howard--Gupta equation (\ref{eqn:howard-grupta}) dominate, and a local radial wavenumber is defined as:
\begin{equation}
    k_r(r) = k\sqrt{\frac{\Delta(r)}{\Phi(r)^2}},
\end{equation}
which is real in regions where the generalised Rayleigh discriminant $\Delta(r)$ is positive:
\begin{equation}
    I \equiv \{r: \Delta(r) \equiv 2\zeta(r)\Omega(r)-\Phi(r)^2 > 0\}.
\end{equation}
The eigenvectors are thus oscillatory in these regions, while evanescent elsewhere. \citet{dizes_2005_2} demonstrated that the neutral modes are localised in $I$ and selected by a discretisation condition that leads to quantised dispersion relations, whose derivative can be approximated as an integral over the localised region:
\begin{equation}
\begin{aligned}
    \pdv{D}{\omega} &= \pdv{}{\omega}\int_I \sqrt{\frac{\Delta}{\Phi^2}}dr \\
    &= -\int_I\frac{\Delta+\Phi^2}{\Delta^{\frac{1}{2}}\Phi^2}\sgn[\Phi]dr.
    \label{eqn:wkbj}
\end{aligned}
\end{equation}
The integrand in (\ref{eqn:wkbj}) shows that the sign of $\partial_\omega D$ depends solely on the sign of $\Phi$ in the oscillatory region. For the smooth neutral modes considered here, $\Phi(r)$ retains a uniform sign throughout $I$: for the regular modes, this is guaranteed by the absence of critical layers; for the discrete critical layer modes, the zero-crossing of $\Phi$ occurs outside the region of integration. Furthermore, since $\Phi$ represents the Doppler-shifted frequency relative to the local flow, which is frame-independent, the sign of $\partial_\omega D$ remains invariant regardless of the observing frame.


We now demonstrate that regular Kelvin waves cannot form explosive triads by using (\ref{eqn:wkbj}) and exploiting its frame invariance. Because this proof fundamentally relies on the dispersion slopes derived from the WKBJ approach, the mathematical guarantee inherits the asymptotic boundaries of the method.
\begin{corollary}
In the large-$k$ limit, regular modes (Kelvin waves) of a columnar vortex alone cannot form explosive triads.
\end{corollary}
\begin{proof}
    By setting proper azimuthal rotation speed and $z$-directional translation speed, we can always position the continuous spectrum of the base flow such that it straddles zero:
    \begin{equation*}
        \min[\Omega(r)]\leq 0 \leq \max[\Omega(r)]
        \quad \mbox{ and\ }\quad
        \min[V_z(r)]\leq 0 \leq \max[V_z(r)],
    \end{equation*}
    which implies that for any set of wavenumbers ($m$ and $k$), the local flow frequency range contains zero:
    \begin{equation}
        \min[m\Omega+kV_z]\leq 0 \leq \max[m\Omega+kV_z].
        \label{eqn:mOmegakV}
    \end{equation}
    Regular modes, by definition, possess eigenfrequencies $\omega$ that lie strictly outside the continuous spectrum (the range of the background fluid frequencies) to avoid critical layers. Therefore, in this frame, the Doppler-shifted frequency $\Phi(r)$ must satisfy:
    \begin{equation*}
        \Phi(r) < 0 \quad (\implies \omega < 0)
        \quad \mbox{or} \quad
        \Phi(r) > 0 \quad (\implies \omega > 0)
    \end{equation*}
    uniformly across the entire domain. This implies that the frame frequency $\omega$ and the intrinsic frequency $\Phi$ share the same sign for every regular mode.

    Recall that the pseudoenergy is proportional to $\omega \cdot \partial_\omega D$. From (\ref{eqn:wkbj}), we know that $\sgn[\partial_\omega D] = -\sgn[\Phi]$. Thus, all regular Kelvin waves possess positive pseudoenergy in this reference frame. Since the nature of an instability (bounded or explosive) must be frame-invariant, and we have found a frame where all pseudoenergies are positive, regular modes cannot form explosive triads.
\end{proof}

The prohibition of explosive instability extends to the broader class of smooth neutral modes. This generalisation follows directly from the resonant structure of the Doppler-shifted frequencies. 
\begin{corollary}
    In the large-$k$ limit, explosive triadic resonance between smooth neutral modes is prohibited, except for the singular case where all three modes are static ($\omega_0 = \omega_1 = \omega_2 = 0$).
\end{corollary}
\begin{proof}
    Explosive instability requires a triad whose interaction coefficient signs follow the relation (\ref{eqn:sign-of-j}). Assume, for the sake of contradiction, that a smooth triad exists with the first two modes sharing the same sign, $\sgn[J_0]\cdot\sgn[J_1]>0$. By (\ref{eqn:sign-of-j-dispersion}) and the property that $\sgn[\partial_\omega D] = -\sgn[\Phi]$ in the WKBJ approximation, this implies:
    \begin{equation}
        \sgn[\omega_0+m_0\Omega+k_0V_z] = \sgn[\omega_1+m_1\Omega+k_1V_z].
        \label{eqn:412}
    \end{equation}
    According to the resonance condition in (\ref{eqn:resonance}), the frequency of the third mode is the sum of the first two. Since $\Phi$ is linear in $\omega, m, k$, we obtain:
    \begin{equation}
    \begin{aligned}
        \sgn[\omega_2+m_2\Omega+k_2V_z] &= \sgn[(\omega_0+\omega_1)+(m_0+m_1)\Omega+(k_0+k_1)V_z] \\
        &= \sgn[(\omega_0+m_0\Omega+k_0V_z) + (\omega_1+m_1\Omega+k_1V_z)].
        \label{eqn:413}
    \end{aligned}
    \end{equation}
    Thus, $\sgn[\Phi_2] = \sgn[\Phi_0] = \sgn[\Phi_1]$. This indicates $\sgn[J_0]=\sgn[J_1]=\sgn[J_2]$, which corresponds strictly to bounded resonance.
\end{proof}

Overall, the prohibition of explosive instability between smooth neutral modes is formally established for the asymptotic regime where the WKBJ approximation remains valid (the large-$k$ limit and low-order azimuthal modes). Within this domain, the only theoretical exception is the singular scenario identified by \citet{saffman_1975}, where all three interacting modes have zero frequency ($\omega = 0$), constituting a specific limiting problem outside the scope of the selection rules derived here. However, the physical implications of these selection rules extend beyond the strict analytical boundaries. As demonstrated by \citet{dizes_2005_2} for monotonic profiles like the Lamb--Oseen vortex, the WKBJ theory provides a good approximation of the dispersion relations for low-$m$ modes across almost all values of $k$. Thus, we conclude that isolated vortices, particularly the monotonic profiles used in aeronautical applications to model aircraft wake vortices, are broadly robust against explosive triadic resonance between smooth neutral modes across a wide, practical range of wave parameters.

The selection rules, however, do not preclude instability if external or internal symmetry-breaking is introduced. Examples include parametric instability driven by external forcing, or the excitation of critical layer singularities where the phase uniformity required for the selection rules no longer hold.


\subsection{Parametric instability of columnar vortices}\label{sec:tuning}
Specific instances of the parametric instability have been studied extensively in the past, most notably the elliptical instability driven by an external strain field \citep{waleffe_1990, kerswell_2002, dizes_2005, dizes_2007}. These studies focus on the particular resonant configurations dictated by the flow geometry or confinement constraints. Our objective here is to establish a general methodology that admits pumping waves of arbitrary frequency and geometry, and identifies resonant triads without relying on the underlying symmetries of the system.

\begin{table}
\centering
\resizebox{\textwidth}{!}{%
\begin{tabular}{cccccc}
Instability                             & Vortex model         & Axial flow          & $\bar{m}$ & $(m_1,m_2)$         & Reference\\ 
\hline\hline
\multirow{8}{*}{Elliptical instability} & Circular vortex      & -                   & $2$       & $(-1,1)$            & \citet{saffman_1975} \\
                                        & Rankine vortex       & -                   & $2$       & $(-1,1)$            & \citet{tsai_widnall_1976} \\
                                        & Burger's vortex      & -                   & $2$       & $(-1,1)$            & \multirow{2}{*}{\citet{dizes_1999}}\\
                                        & Lamb--Oseen vortex    & -                   & $2$       & $(-1,1)$            &   \\
                                        & Gaussian vortex      & -                   & $2$       & $(-1,1)$            & \citet{dizes_2002} \\
                        & \multirow{2}{*}{Rankine vortex}      & constant axial flow & \multirow{2}{*}{$2$} & \multirow{2}{*}{$m_2-m_1 = 2$} & \multirow{2}{*}{\citet{dizes_2005}}  \\
                                        &                      & inside vortex core  &           &                     &   \\
                                        & Batchelor vortex     & jet-like axial flow & $2$       & $m_2-m_1 = 2$       & \citet{dizes_2007}\\
                                        & Moore-Saffman vortex & jet-like axial flow & $2$       & $m_2-m_1 = 2$       & \citet{feys_maslowe_2016}\\ 
\midrule
Multi-polar strain                      & Rankine vortex       & -                   & $2,3,4$   & $m_2-m_1 = \bar{m}$ & \citet{dizes_2001}\\ 
\midrule
\multirow{2}{*}{Curvature instability}  & \multirow{2}{*}{Batchelor vortex}          & \multirow{2}{*}{jet-like axial flow}                                     & \multirow{2}{*}{$1$} & \multirow{2}{*}{$m_2-m_1 = 1$} & \multirow{2}{*}{\begin{minipage}[t]{0.25\columnwidth}\centering\citet{dizes_2017}\end{minipage}}\\
          &                      &                                &                   &           &           \\ 
\midrule
Radial perturbation                     & Lamb--Oseen vortex    &  -                  & $0$       & -                 & \citet{maslowe_2022}\\
\bottomrule
\end{tabular}}
\caption{Summary of resonant triads examined in the literature. A blank axial flow entry indicates a vortex with no axial velocity component. In all listed cases, the perturbation modes have no axial dependency ($\bar{k}=0$) and the base wave eigenvalues are either zero or assumed zero ($\bar{\sigma}=0$), except for the Gaussian vortex case. Note that the radial perturbation at the bottom falls into the $(\bar{m},\bar{k}) = (0,0)$ self-interaction case that is not considered in this study.}
\label{tab:triads}
\end{table}

\begin{figure}
    \centering
    \includegraphics[width=\linewidth]{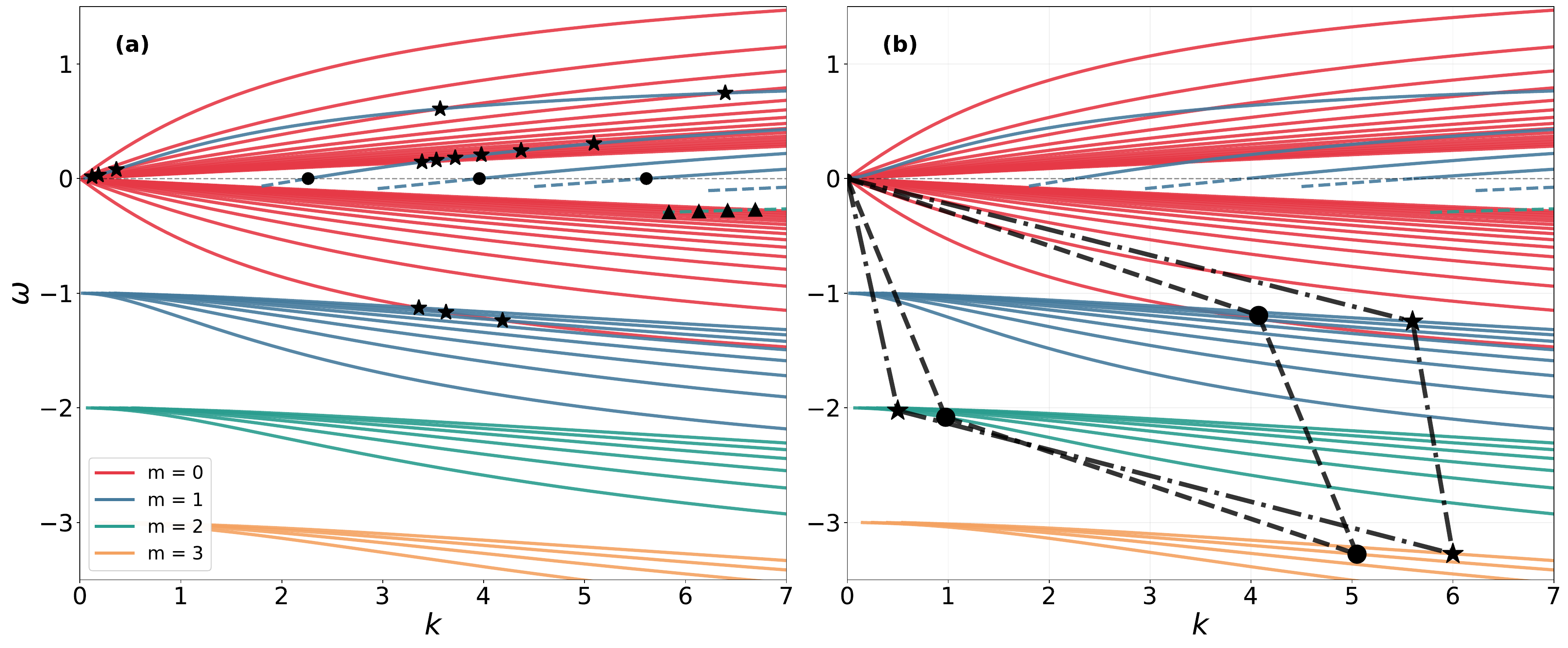}
    \caption{Graphical representations of resonant triads in the Lamb--Oseen vortex. (a) Resonance sustained by a stationary pumping wave with zero axial wavenumber ($m_p, k_p=0,\omega_p=0$). Crossing points between distinct dispersion branches ($k_i = k_j$, $\omega_i = \omega_j$) are marked by stars for $|m_i - m_j| = 1$ and solid triangles for $|m_i - m_j| = 2$, representing resonance with an external field of symmetry $|m_p| = |m_i - m_j|$. Filled circles indicate intersections with the zero-frequency axis, corresponding to a pair of conjugate modes ($\pm m$) in resonance with a static strain field with symmetry $m_p=2m$. (b) General triadic resonance satisfying the frequency matching condition $\omega_1 + \omega_2 = \omega_3$ alongside wavenumber matching. The specific examples represent resonance between modes $(m,k,\omega)$: $(1, 5.6, -1.2) + (2, 0.50, -2.0) = (3, 6.0, -3.3)$ (dash-dotted lines with star markers), and $(1, 4.1, -1.2) + (2, 0.98, -2.1) = (3, 5.1, -3.3)$ (dashed lines with circular markers).}
    \label{fig:dispersion}
\end{figure}

Table~\ref{tab:triads} summarises the resonant configurations examined in the literature. 
In the classical elliptical instability, the forcing arises from an external strain field imposed by the counter-rotating vortex. Because this strain transmits into the vortex core to induce steady elliptical streamlines \citep{moffatt_1994}, the interaction can be modelled mathematically as a stationary quadrupolar deformation ($m_0 = 2, \omega_0 = 0$). 
Similarly, studies on the precessional instability often utilise the precessing reference frame to render the dipolar forcing stationary ($m_0 = 1, \omega_0 = 0$). In these scenarios, resonant triads can be located at discrete intersections of the dispersion curves, which were initially identified by \citet{saffman_1975} at the zero-crossings of the dispersion curves via the conjugate symmetry of the linear eigenvalues, as highlighted by solid round marks in figure~\ref{fig:dispersion}a for the Lamb-Oseen vortex. This was later extended to the general cross points by numerous authors \citep[as catalogued in table~\ref{tab:triads}; see e.g.,][]{dizes_1999, Hussain_2001, kerswell_2002, fukumoto_2005}, as indicated in the same figure by filled triangles for the elliptical instability and solid stars for both the precessional and curvature instabilities. On the contrary, when the pumping parameters are treated as free variables, the task of identifying resonant triads is no longer trivial. Schematically, the resonant triads form a parallelogram in the dispersion space, as shown in figure~\ref{fig:dispersion}b: the vectors $(k_0,\omega_0)$ and $(k_1,\omega_1)$ define the adjacent sides, while $(k_2,\omega_2)$ completes the vertex. Since the vertices can move freely along the dispersion curves, a countless number of resonant triads are permitted. 

\subsubsection{Resonance tuning via non-degenerate perturbation theory}\label{sec:tuning-method}
For the general case, if the dispersion relations are known \textit{a priori}, identifying resonant triads is geometrically straightforward: one can simply shift the origin of the dispersion curve of a constituent mode along the dispersion curve of the pump mode until an intersection with the dispersion curve of the remaining mode is identified. This shifting is equivalent to adjusting the relative axial wavenumbers of the modes. However, mapping the full dispersion landscape can be computationally expensive for complex base flow profiles. Furthermore, when a strict numerical tolerance for perfect resonance is required, it is more efficient to use non-degenerate perturbation theory to iteratively tune the axial wavenumbers of an arbitrary triad. Specifically, since the dispersion relation varies continuously with the axial wavenumber $k$, we enforce the exact frequency resonance condition ($\Delta\omega=0$) by performing a linear expansion of the frequency mismatch with respect to a small wavenumber correction $\Delta k$, and setting the resulting expression to zero:
\begin{equation}
    \Delta\omega(k+\Delta k) \approx \Delta\omega(k) + \left( \pdv{\omega_2}{k} - \pdv{\omega_1}{k} \right)\Delta k = 0.
\end{equation}
By iteratively solving for $\Delta k$, we drive $\Delta\omega \to 0$, thereby locking the system into perfect resonance where the growth rate is maximised. 

\begin{figure}
    \centering
    \includegraphics[width=0.75\linewidth]{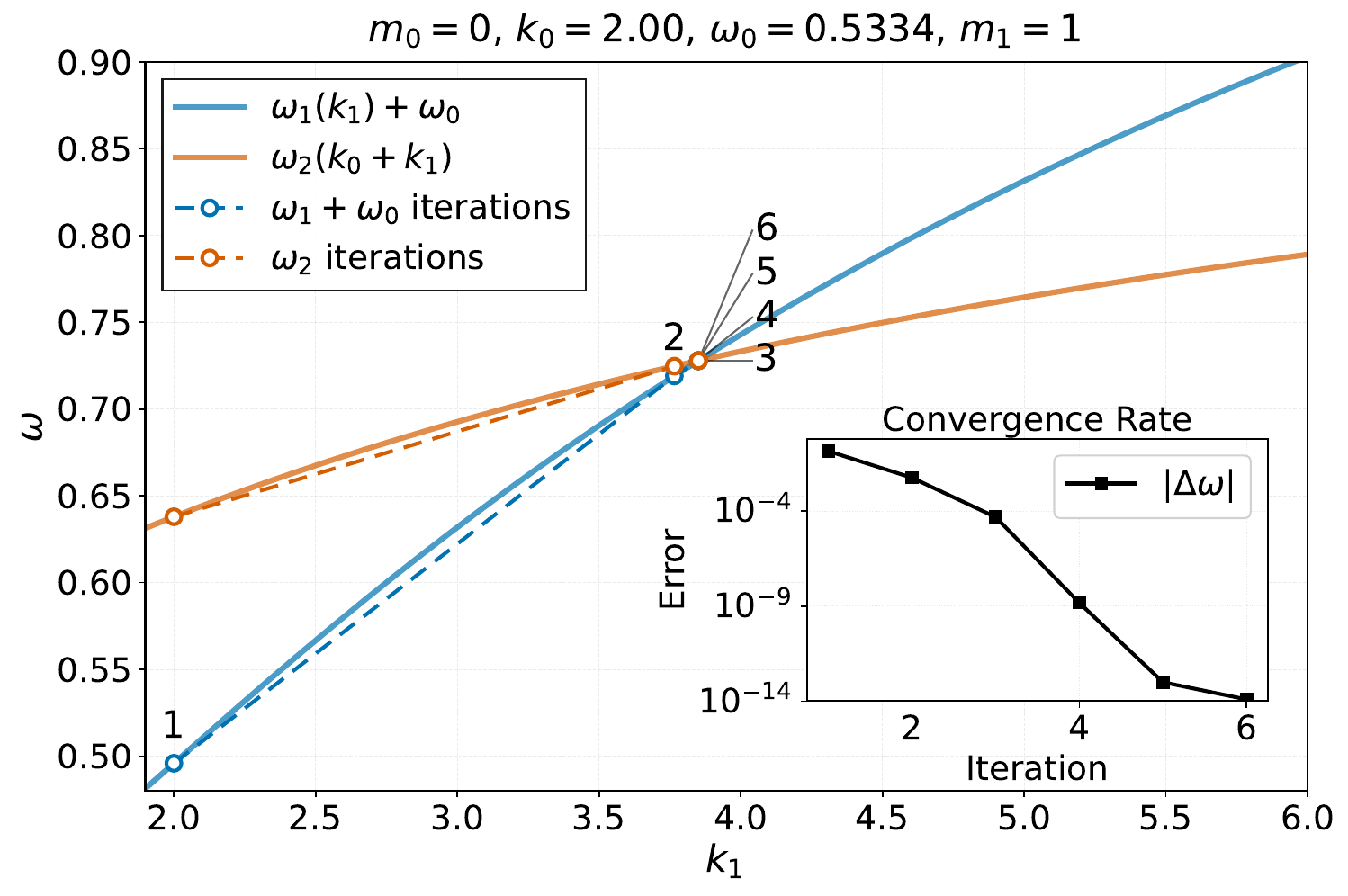}
    \caption{Iterations of the resonance tuning process for a triad driven by the pump mode $(m_0, k_0, \omega_0) = (0, 2.00, 0.5334)$ of the Lamb--Oseen vortex. Note that iterations $3$-$6$ are overlapping as the frequency mismatch approaches machine precision during the tuning. Solid lines are the dispersion curves of the two free modes, $\omega_0 + \omega_1(k_1)$ and $\omega_2(k_1+k_0)$. They are provided and shifted here for visual clarity so that their intersection corresponds to perfect resonance. Dashed lines trace the trajectories of the perturbation-based Newton--Raphson algorithm as it projects from one wavenumber approximation $k_1^{(n)}$ to the next $k_1^{(n+1)}$ based on the local group velocity mismatch given by (\ref{eqn:tuning}). The inset (lower right) quantifies the performance, showing the frequency mismatch $|\Delta \omega|$ reducing to a precision of $\order{10^{-1}}$ during the first iteration, and to $\order{10^{-14}}$ after 6 iterations.}   
    \label{fig:tuning}
\end{figure}

In the numerical implementation, we apply a finite test shift $\delta k$ to the axial wavenumbers of the free modes ($k_j'=k_j + \delta k$) and compute the resulting eigenvalue shift $\delta\sigma_j$ via the first-order correction from non-degenerate perturbation theory:
\begin{equation}
    \delta\sigma_j = \left\langle \Tilde{\bm{L}}_j\middle| \left(\mathbb{M}_{m_jk_j'}-\mathbb{M}_{m_jk_j}\right)\Tilde{\bm{R}}_j\right\rangle.
\end{equation}
The wavenumber correction required for perfect resonance is then estimated as:
\begin{equation}
    \Delta k \approx \delta k\cdot\frac{\sigma_1+\sigma_0-\sigma_2}{\delta\sigma_2-\delta\sigma_1} .
    \label{eqn:tuning}
\end{equation}
Figure~\ref{fig:tuning} illustrates this iterative tuning process, which constitutes a Newton--Raphson root-finding scheme capable of achieving arbitrary precision. As shown in the figure, the method is highly efficient as the frequency mismatch $|\Delta \omega|$ typically drops to $\order{10^{-3}}$ after a single iteration. Note that the shifted dispersion curves in figure~\ref{fig:tuning} are only for visualisation purposes, and the actual tuning process does not require the full mapping of the dispersion relations, hence suitable for arbitrary base flow profiles. Overall, this tuning capability generalises the stability analysis, allowing us to locate parametric instabilities across a broad spectrum of pumping frequencies limited only by the existence of a valid triadic resonance in the parameter space.

While the number of resonant triads is effectively countless, resonance does not necessarily guarantee instability as discussed in \S\ref{sec:conservative-triad}. Relating the interaction coefficients in (\ref{eqn:corrected-parametric-sigma}) to the pseudoenergy signs examined in \S\ref{sec:selection-rules} shows that the parametric instability necessitates a negative energy mode (specifically, a critical layer mode), except for the singular limit of zero-frequency modes. This requirement is satisfied by the mixed-sign subclass of the bounded regime identified in \S\ref{sec:conservative-triad}, where the sum mode shares an energy sign with the pump mode (i.e., $s_2 = s_0 = -s_1$). Although such configurations are topologically prohibited from forming explosive resonance, the external forcing of the pump wave effectively bypasses the Manley--Rowe constraints that would otherwise strictly limit the interaction amplitude. This energetic mechanism aligns exactly with the findings of \citet{dizes_2008}, who recognised critical layer modes as a prerequisite for the elliptical instability. 

\begin{table}
\centering
\begin{tabular}{cccc}
\multicolumn{2}{c}{\citet{dizes_1999}} & \multicolumn{2}{c}{Present study}  \\ 
\cmidrule(lr){1-2}\cmidrule(lr){3-4}
$k$   & $\Re[\sigma_{\text{par}}]/\epsilon$                   & $k$       & $\Re[\sigma_{\text{par}}]/\epsilon$  \\
2.26 & 1.3790                        & 2.260814 & 1.379248       \\
3.94 & 1.3888                        & 3.957585 & 1.388915       \\
-  & -                               & 5.611789 & 1.391302       \\
-  & -                               & 7.252515 & 1.392252       
\end{tabular}
\caption{Comparison of the axial wavenumbers ($k$) and the parametric growth rates computed for the elliptical instability of the Lamb--Oseen vortex between a pair of helical ($m=\pm 1$) wave modes with equal $k$. These modes are the zero-crossings of the dispersion curves ($\omega=0$), highlighted as solid round marks in figure~\ref{fig:dispersion}a. Calculation of the stationary pump mode (the strain field correction) is detailed in Appendix~\ref{appendix:mlegs}, and the numerical values of the growth rate in this work are calculated based on (\ref{eqn:corrected-parametric-sigma}) using MLEGS. Note that \citet{dizes_1999} only computed the parametric growth rate for the first two degenerate pairs, and their results are scaled by $2$ here due to a scaling in the radial direction for the pump mode.}
\label{tab:elliptical}
\end{table}

\subsubsection{Parametric instability of the Lamb-Oseen vortex}
As a demonstration, we perform a numerical investigation of the parametric instability in the Lamb--Oseen vortex. To confirm the numerical validity of our framework, we compare our numerical results based on (\ref{eqn:corrected-parametric-sigma}) against the classical elliptical instability results calculated by \citet{dizes_1999}. In the work by \citet{saffman_1975}, the resonant modes are located using the azimuthal symmetry of the linear dispersion relations, which suggests a pair of free wave modes being in resonance with the imposed strain mode when $m_2 = -m_1 = 1$ and $\omega_1 = \omega_2 = 0$. These modes can be found as the zero-roots of the linear dispersion relations, which are highlighted by solid round marks in figure~\ref{fig:dispersion}a. The external strain field associated with the elliptical deformation is obtained as a stationary mode ($m_0 = 2, k_0 = 0, \omega_0 = 0$) with a distinct, non-decaying boundary condition at the far field as detailed in Appendix~\ref{appendix:mlegs}. Table~\ref{tab:elliptical} lists the computed growth rates for various resonant pairs, which are in excellent agreement with the numerical results of \citet{dizes_1999} to four significant figures. We further extend the calculation to higher wavenumbers, which were not explored in the previous work, and find that the growth rate continues to increase with $k$ within our search range, albeit at a diminishing rate.

Moving beyond the classical elliptical instability, we look for general parametric instability of the Lamb--Oseen vortex. We note that the resonance condition allows for mode swapping. If a triad is found where the free modes are both of the opposite kind to the pump mode (e.g., a critical layer mode as the pump mode, and two regular modes as the free modes), which yields no parametric instability, one can swap the pump mode $m_0$ with free mode $m_1$ to create a new triad configuration. In this new arrangement, the free modes now have opposite pseudoenergies (e.g., regular and critical layer), thereby unlocking the parametric instability. In other words, as long as the pump mode and the first free mode are of opposite pseudoenergy signs, they always possess a parametric instability configuration regardless of the kind associated with the sum mode.

\begin{figure}
    \centering
    \includegraphics[width=\linewidth]{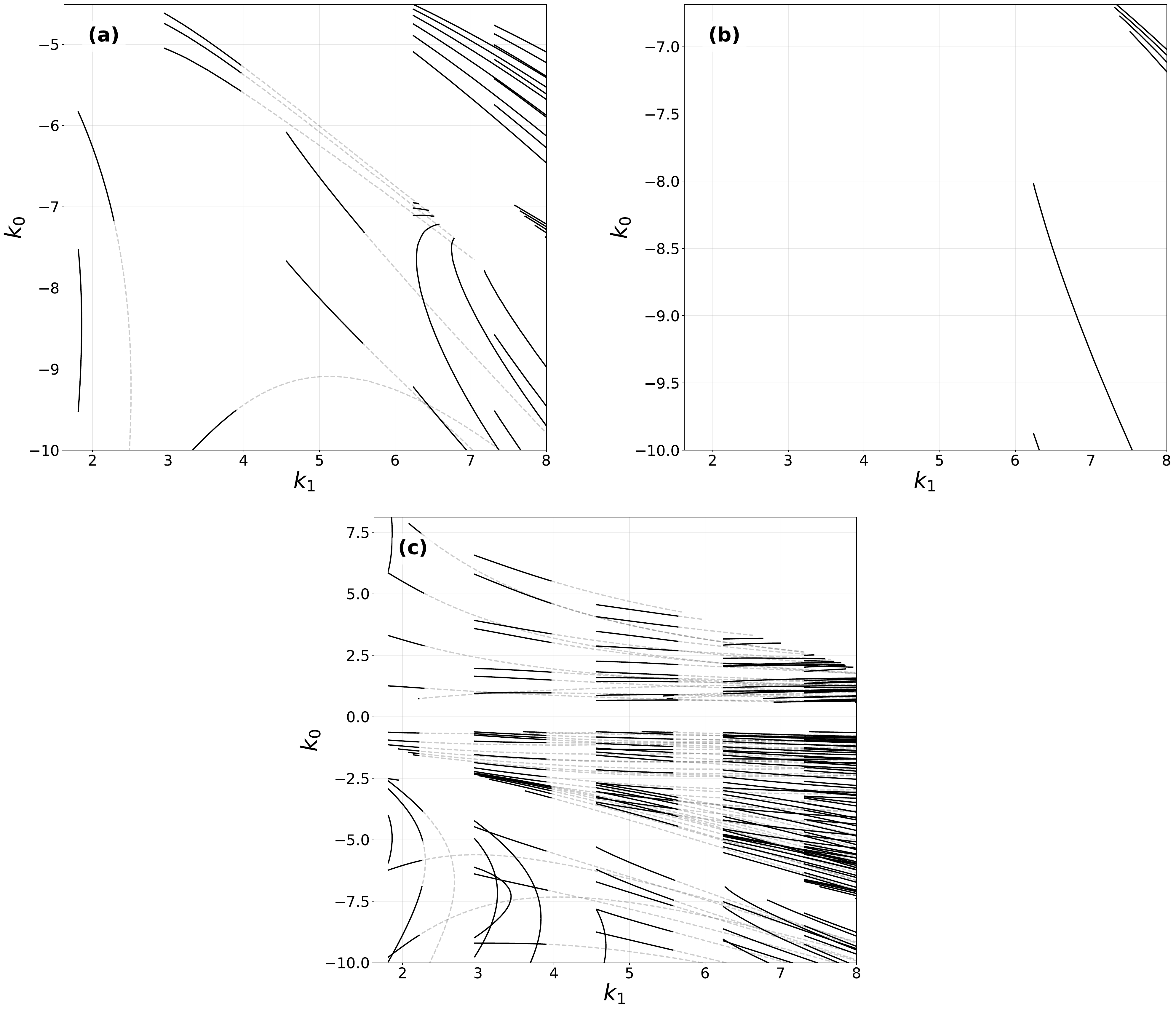}
    \caption{Resonance configurations in the $(k_1, k_0)$ plane for generalised parametric instability. The triads involve a helical ($m_1=1$) discrete critical layer mode and an intrinsic pump mode with azimuthal wavenumber of (a) $m_0 = 1$; (b) $m_0 = 2$; (c) $m_0 = 0$. Any point along these curves determines the geometric configuration of the resonant triad, with the third mode satisfying $k_2 = k_0 + k_1$ and $m_2 = m_0 + m_1$. The solid lines trace continuous families of configurations that trigger parametric instability based on the pseudoenergy criterion. The grey dashed lines represent extensions where the discrete critical layer mode transitions into a regular Kelvin wave, thereby suppressing the instability. All modes evaluated are intrinsic normal modes satisfying decaying far-field boundary conditions. Consequently, classical forced instabilities summarised in table~\ref{tab:triads} are excluded; if plotted, they would manifest as discrete points (along the $k_0 = 0$ axis) rather than continuous curves. The search spans the neutral branches resolved in figure~\ref{fig:dispersion}, except for $m_0 = 0$ where only the first 6 radial branches are considered for visual clarity.}
    \label{fig:cl_branches}
\end{figure}

Following this logic, we perform a systematic search for parametric instability by considering all possible triads where the first two modes possess opposite pseudoenergy signs. To explore the true breadth of this phenomenon, we focus our automated search entirely on intrinsic normal modes that satisfy standard decaying boundary conditions at the far field, which can be directly evaluated from MLEGS without separate treatment. This approach explores the generalised parametric instability among modes native to the isolated vortex, fundamentally distinguishing our search from the classical elliptical instability setup where the pump mode models an externally imposed strain field and requires the aforementioned non-decaying boundary condition. Although our mathematical formulation applies to both scenarios, our general search is significantly broader and more complicated, as the classical pump mode is typically assumed to be two-dimensional and stationary ($k_0 = 0, \omega_0 = 0$), which vastly simplifies the tuning process. Based on the pseudoenergy signs established earlier in the WKBJ framework, this search is effectively reduced to scanning the parameter space of the pump mode $m_0$ while keeping the first free mode $m_1$ fixed as a discrete critical layer mode. The search focuses on the discrete neutral branches as resolved in figure~\ref{fig:dispersion} for each azimuthal wavenumber, except for $m_0 = 0$ where we consider the first $6$ radial branches due to the dense spectral spacing between higher-order branches.

Within the search range of 
$0 \leq|k|\leq 10$, we scan for triads involving pumping azimuthal wavenumbers of $m_0 = 0, 1, 2$ and discrete critical layer modes of azimuthal wavenumber $m_1 = 1$. This comprehensive search locates numerous parametric instability families, which are visualised in figure~\ref{fig:cl_branches}. To properly navigate this stability landscape, it is important to understand how these resonance configurations are visualised. Specifically, each panel maps the axial wavenumber of the free critical layer mode ($k_1$) against that of the intrinsic pump mode ($k_0$). Because the azimuthal wavenumbers are pre-specified for each subplot, locating any point along these curves immediately determines the geometric configuration of the entire resonant triad based on the resonance condition (\ref{eqn:resonance}). Furthermore, unlike classical configurations (e.g., the elliptical instability configurations as shown in figure~\ref{fig:dispersion}a and validated in table~\ref{tab:elliptical}) which would merely manifest as discrete points along the $k_0 = 0$ axis, the located resonance configurations trace out dense, continuous curves as an inherent result of our generalised framework: not only are the axial wavenumbers continuously adjustable, but we also place no limitation on the temporal frequency, unlocking entire families of resonances that a strictly stationary assumption would miss. The solid line segments in figure~\ref{fig:cl_branches} trace the continuous families that trigger parametric instability, while the grey dashed continuations represent extensions where varying the wavenumbers causes the discrete critical layer mode to transfer to a regular Kelvin wave, thereby suppressing the instability. The parametric growth rates of the located unstable configurations are computed numerically from their interaction coefficients based on (\ref{eqn:corrected-parametric-sigma}) and are all found to be real, confirming their instability as indicated by the pseudoenergy criterion. Note that by isolating intrinsic modes with decaying boundaries, we inherently exclude the classical configurations from this plot.

It is worth pointing out that, because of the conjugate and azimuthal symmetries of the underlying eigenvalue problem, the existence of parametric configurations involving $m_1 = 1$ implies the existence of a symmetric set of $m_1 = -1$, and the corresponding configurations are not presented here to avoid redundancy. Additionally, while we did not locate any general intrinsic parametric configurations involving $m_1=2$ discrete critical layer modes within the specific $|k|\leq 10$ search range, we expect such configurations to exist at higher wavenumbers due to the limited span of their neutral branches at lower wavenumbers (see figure~\ref{fig:dispersion}).

We also emphasise that our numerical search is strictly restricted to smooth neutral modes, which are verified by monitoring their spectral expansion coefficients based on the spectral convergence guaranteed by the MLEGS basis functions for smooth profiles. This distinguishes our work from \citet{jacquin_2006} and \citet{dizes_2005_2}, where dispersion curves of discrete critical layer modes are extended to much lower wavenumbers. As $|k|$ decreases, the critical layer moves closer to the vortex core, the singularity becomes increasingly prominent, and the modes lose their global phase coherence. Consequently, they can no longer form conservative resonances, which is why we exclude these lower-wavenumber extensions. Similarly, the singular eigenfunctions of the continuous spectrum are entirely excluded because they too do not form conservative resonances. However, this does not imply that the critical layer itself cannot participate in nonlinear interactions; relevant discussions regarding critical layer dynamics in non-conservative frameworks are provided in \S\ref{sec:critical-layer}.

Ultimately, while we do not claim to have exhausted the entire parameter space, the continuous resonance configurations in figure~\ref{fig:cl_branches} demonstrate that general parametric instability is effectively countless and spans broad continuous bands. Most importantly, we have successfully applied our numerical capabilities to practically realise the theoretical guidance derived from the pseudoenergy criterion. By systematically locating these continuous families of unstable triads via the tuning method described in \S\ref{sec:tuning-method}, we demonstrate that analytical insights from conservative triadic resonance can be actively and systematically translated into real, tunable instabilities for arbitrary pumping frequencies and wavenumbers, an approach which, to the best of our knowledge, has not been conducted before.

\section{Roles of the critical layer in the instability of columnar vortices}\label{sec:critical-layer}

The selection rules derived in \S\ref{sec:selection-rules} establish that isolated columnar vortices are robust against explosive resonance when the interacting waves remain smooth and neutral. This constraint applies not only to regular Kelvin waves but also to the discrete critical layer modes, whose critical layer singularity is effectively passive. Consequently, the onset of intrinsic instability requires a mechanism that violates the conservation laws imposed by the Manley--Rowe relations. External forcing, quantified in the preceding section, constitutes the first such mechanism. This section examines the second pathway, which centres on the critical layer itself. In contrast to the rigorous perturbation analysis employed earlier, the discussion here is qualitative and serves two purposes. First, we explain the physical origin of the negative pseudoenergy carried by waves possessing critical layers, complementing the WKBJ-based assignment adopted in \S\ref{sec:selection-rules}. Second, we distinguish the passive critical layers of the discrete modes from \textit{active} critical layers, whose strong wave-mean interaction extracts energy directly from the background shear and transforms the otherwise conservative system into an open one potentially capable of intrinsic instability.

To understand the physical mechanism of the critical layer, we start by introducing the energy budget of the disturbance field. The time rate of change of the total disturbance kinetic energy, $KE$, is governed solely by the shear production term:
\begin{equation}
\begin{aligned}
    \dv{t} KE &= - \iiint_V \bm{u} \cdot (\bm{u} \cdot \grad \vq) dV \\
    &= -\iiint_V \left(u_r u_\phi\cdot r\dv{r} \Omega + u_r u_z\cdot \dv{r} V_z\right) dV,
\end{aligned}
\label{eqn:energy_balance}
\end{equation}
where the two terms in the integral represent the work done by the disturbance against the rotational shear ($r\dv{r} \Omega$) and the axial shear ($\dv{r} V_z$), respectively. This relation indicates that, without external forcing (which is the mechanism for the parametric instability), a disturbance must be able to tap into the energy of the background shear via wave-mean interactions in order to sustain its growth.

In this regard, the critical layer is an efficient pathway to tap into this background energy via wave-mean resonance (i.e., a local resonance between the wave and the background shear flow). Fundamentally, a critical layer is a physical locus where the wave phase velocity matches the local flow velocity, suggesting a zero Doppler-shifted frequency, $\Phi(r_c) = 0$. In the linear inviscid theory governed by the Howard--Gupta equation (\ref{eqn:howard-grupta}), this wave-mean resonance manifests mathematically as a singularity, leading to a logarithmic divergence in the vorticity fluctuation and essentially separating the eigenmodes into disjoint pieces. However, this singularity is strictly a result of mathematical idealisation. In physical flows, the singularity is inherently regularised by physical mechanisms that are overlooked by the inviscid linear limit, namely viscosity and nonlinearity. As widely documented in the study of planar shear flows, depending on the relative strength of these different mechanisms, the physical system can exist anywhere along a spectrum from viscosity-dominated regularisation to nonlinearity-dominated regularisation \citep{benney_bergeron_1969, warn_1976, warn_1978, Stewartson_1977}. In what follows, we first discuss the viscosity-dominated regime, which underlies our weakly nonlinear formulation, and defer the nonlinearity-dominated regime to the end of this section.

Our weakly nonlinear formulation inherently follows the viscosity-dominated description. In the multi-scale framework of \S\ref{sec:multi-scale}, the triadic coupling is quantified using the leading-order wave structures, which obey the linear equation (\ref{eqn:howard-grupta}) and must be connected across the critical radius in the inviscid limit through the classical diminishing-viscosity prescription \citep{lin_1955, dizes_2004}. Under this prescription, the solution acquires a phase jump, typically $-\pi$, which embodies the torque that the wave exerts on the mean flow at the critical radius. The capacity of a wave to feed on this exchange is quantified by its \textit{pseudoenergy}, defined as the net change in the total energy of the flow required to support the disturbance. A positive pseudoenergy wave requires a net energy input to grow, whereas a negative pseudoenergy wave lowers the total energy of the system by withdrawing energy from the supporting background flow. The WKBJ analysis in \S\ref{sec:selection-rules}, which adopts the same prescription \citep{dizes_2005_2}, identifies the waves possessing critical layers as the sole carriers of negative pseudoenergy among the smooth neutral modes. This identification is natural because the critical layer is the location where a wave can tap the background energy in (\ref{eqn:energy_balance}). Following this, a negative pseudoenergy wave is the essential ingredient of explosive resonance because it permits the simultaneous amplification of an entire triad without violating the total energy budget. As established in \S\ref{sec:selection-rules}, however, smooth neutral modes never realise the configuration required for explosive growth due to the selection rules, which are the constraints that an \textit{active} critical layer may circumvent.

What separates a passive critical layer from an active one is the wave amplitude at the critical radius. For the discrete critical layer modes, the oscillatory region of the eigenmode is confined near the vortex core and separated from the critical radius by a wide evanescent region \citep{dizes_2005_2}. The eigenmode therefore maintains phase coherence throughout the region where it possesses appreciable amplitude, and its exponentially small amplitude at the critical radius suppresses the singularity together with the torque exerted there. The critical layers of these discrete modes are thus passive, and the modes behave as ordinary neutral waves that exchange no net energy with the mean flow at the linear level while remaining nearly neutral under finite viscosity. Their negative pseudoenergy nevertheless remains encoded in their global structure, and it takes effect during the nonlinear wave coupling, as in the parametric configurations of \S\ref{sec:selection-rules}, where the growth of a negative pseudoenergy wave compensates the energy gained by its positive pseudoenergy partner \citep{cairns_1979}. In contrast, when the critical radius lies within the region of appreciable wave amplitude, as is the case for the singular eigenfunctions of the continuous spectrum, the critical layer is active. The strong wave-mean interaction there breaks the global phase coherence of the leading-order wave structure in the inviscid limit, which is why they are excluded from our preceding analytical formulation.

We now turn to the nonlinearity-dominated regime. At higher wave amplitudes, nonlinearity overtakes viscosity as the dominant regularisation mechanism, and the situation becomes particularly delicate in the presence of ``stacked'' resonances. Due to the linear relation $\Phi_2(r) = \Phi_0(r) + \Phi_1(r)$ governing a resonant triad, if any two members share a critical layer, the third mode is geometrically constrained to possess a critical layer at the precise same location. Such a configuration is highly conducive to a nonlinear cascade, as the primary modes spawn higher-order harmonics or secondary interactions that remain spatially locked to the same critical radius. The strong, localised nonlinearities then violate the standard asymptotic ordering, so our weakly nonlinear multi-scale perturbation framework no longer applies and a dedicated nonlinear critical layer theory is required \citep{maslowe_1986}. While such theories are well-established for planar geometries \citep{benney_bergeron_1969,Stewartson_1977,warn_1976,warn_1978}, their application to critical layers in columnar vortices remains largely unexplored.

Lastly, although a complete theory of active critical layers is not yet available for columnar vortices, relevant mechanisms have been documented in simpler planar systems. On magnetised electron columns, whose dynamics are isomorphic to the two dimensional Euler equations, a discrete Kelvin (diocotron) wave has been observed to decay into a daughter mode with a lower wavenumber by coupling to a \textit{beat wave}, a disturbance that is not a discrete mode but is resonantly absorbed at its critical radius \citep{mitchell_1994, mattor_2006, dubin_2024}. Whether analogous interactions mediated by critical layers occur in columnar vortices, where the waves additionally carry axial structure, has yet to be examined. Nevertheless, as a fundamental nonconservative mechanism, the possibility that active critical layers might trigger spontaneous instabilities warrants dedicated investigation.

\section{Conclusion and future work}
In this work, we have unified the theory of vortex instability under the framework of hydrodynamic ``selection rules''. By performing a multi-scale perturbation analysis, we derived the general amplitude equations for triadic resonance in columnar vortices. We established that the resonant interaction between smooth neutral wave modes, specifically regular Kelvin waves and discrete critical layer modes with passive singularities, is inherently conservative. Consequently, the dynamics of these triads are governed by the Manley--Rowe relations, which constrain their energy exchange to compact manifolds.

To determine the stability of these conservative interactions, we employed the concept of wave pseudoenergy within a large-$k$ WKBJ framework by \citet{dizes_2005_2}. We demonstrated that the sign of pseudoenergy is invariant under a translating frame and that critical layer modes are the sole carriers of negative pseudoenergy. This leads to a set of selection rules that topologically prohibit the formation of explosive resonant triads. This finding provides a theoretical explanation for the robustness of isolated columnar vortices, showing that intrinsic instability is not a generic property of resonance but is suppressed by the underlying Hamiltonian structure of the flow.

Despite the prohibition of explosive resonance, we showed that the conservative framework admits \textit{parametric instability} when the system is driven by external forcing. In this scenario, the external pump acts as an energy reservoir that bypasses the Manley--Rowe constraints. To explore this, we developed an efficient tuning method based on non-degenerate perturbation theory. Unlike classical approaches, our method places no restrictions on the imposed pump mode, allowing us to actively locate resonant configurations for arbitrary pumping frequencies and wavenumbers. We validated this method on the Lamb--Oseen vortex, reproducing classical results for elliptical instability and identifying a broad class of new parametric instabilities. Specifically, we demonstrated that while the WKBJ approximation is formally restricted to the large-$k$ limit and small azimuthal wavenumbers $m$, utilising it to theoretically determine pseudoenergy signs successfully guides our tuning algorithm to uncover a vast, continuous landscape of new parametric instabilities involving discrete critical layer modes, far exceeding the highly restricted configurations of past studies.

We further identified \textit{active critical layers} as an alternative pathway to bypass the selection rules. Unlike the passive critical layers of the discrete modes, an active critical layer sustains a strong wave-mean resonance that renders wave interactions non-conservative. Studying triadic resonance that involves active critical layers therefore requires extensions beyond the conservative limit, and the appropriate extension depends on the dominant regularisation. In the viscous regime, the spectrum offers additional families of damped discrete eigenmodes, including the viscous critical layer eigenmodes reported by \citet{lee_marcus}. Such waves are directly accommodated by the general amplitude equations (\ref{eqn:3-wave-general}), which retain the linear growth and decay rates of the participating waves, although, as noted by \citet{craik_1986}, the resulting triadic dynamics is significant only when the linear damping is sufficiently weak. A comprehensive numerical analysis of non-conservative triadic resonance involving these damped waves is currently being conducted. 
On the other hand, when nonlinearity instead dominates the regularisation, the amplitude framework no longer applies, and the evaluation awaits a nonlinear critical layer theory for columnar vortices, as discussed in \S\ref{sec:critical-layer}.

Notably, observations from two dimensional vortex experiments provide a precedent for active critical layers participating in nonlinear interactions, where a Kelvin wave decays into a daughter wave by coupling to a beat wave absorbed at its critical radius \citep{mitchell_1994, mattor_2006, dubin_2024}. Because the beat wave is a continuum response rather than a discrete mode, modeling its columnar analogue would further require generalising one leg of the amplitude equations to such a response, an extension we leave for future work. Ultimately, our findings highlight the importance of specific mechanisms for symmetry breaking in the mitigation of aircraft wake vortices: either tuned external forcing to trigger parametric instability or the excitation of active critical layers.

Beyond the unstratified dynamics explored thus far, the introduction of background density stratification acts as an additional physical mechanism that fundamentally alters the stability landscape. In aircraft wake vortices, phenomena such as exhaust entrainment or wingtip heating can introduce strong radial density variations, whose destabilising potentials have been investigated in various studies \citep[see, for example,][]{joly_2005, sipp_2005, Smith_2021}. In the context of wave interactions, as demonstrated by \citet{Dixit_2011}, a radial density gradient placed near the critical radius profoundly modifies vortex stability through resonance. Specifically, they showed that an instability emerges where the dispersion curves of two otherwise neutral waves intersect: a density wave propagating on a radial interface and a regular Kelvin mode propagating on the vortex core. Most notably, their study employed the exact same pseudoenergy criterion to explain this two-wave resonance. 
Building on this foundation, applying our multi-scale perturbation framework to study triadic resonance under the coupled effects of background shear and continuous radial stratification remains a highly promising avenue for future research.

Finally, while aircraft wake vortex mitigation serves as the primary motivation for this study, columnar vortices and the fundamental rotating shear they provide are highly universal phenomena frequently observed in geophysical and astrophysical contexts. In these broader environments, background stratification is typically aligned along the rotation axis rather than radially, which introduces a new class of baroclinic critical layer modes. Studies have found that both this new class and the traditional class of critical layers play a prominent role in the stability of such flows, actively driving phenomena such as radiative instability \citep{dizes_2009} and Zombie Vortex Instability \citep{Marcus_2013,Marcus_2015,Marcus_2016,Barranco_2018}. The relevant nonlinear critical layer theory and secondary instabilities have been examined by \citet{Balmforth_2020, Balmforth_2021}. Investigating three-wave resonance within these axially stratified regimes presents a compelling direction for future study, highlighting the universal importance of critical layers across various fluid mechanical disciplines.

\section*{Acknowledgements}
We thank Dr. Neil Balmforth and the reviewers for their constructive feedback and helpful literature recommendations.

\section*{Funding statements}
This work used Anvil at Purdue University through allocation PHY220056 from the Advanced Cyberinfrastructure Coordination Ecosystem: Services \& Support (ACCESS) program \citep{boerner_2023}, which is supported by U.S. National Science Foundation grants \#2138259, \#2138286, \#2138307, \#2137603, and \#2138296.

\section*{Competing Interests}
The author(s) declare none.

\section*{Data availability statement}
The solver used for analysis in this study is available as open-source software on Zenodo at https://doi.org/10.5281/zenodo.17095748 \citep{MLEGS}.

\appendix

\section{Linear system in a rotating frame}\label{app:rotation}\label{sec:galilean}
Consider an observing frame rotating with a constant angular velocity $\bar{\Omega}\uvec_z$, a velocity field will appear to have a superposed velocity $-r\bar{\Omega}\uvec_\phi$, and the differential operators in the rotating frame and the original frame are related as follows:
\begin{equation}
    \grad = \grad'
    \quad \mbox{and\ }\quad
    \pdv{}{t} = \pdv{}{t'}-\bar{\Omega}\pdv{}{\phi'},
\end{equation}
where $\{\cdot\}'$ denotes measurements and operators made with respect to the rotating frame. The Euler equations give
\begin{equation}
\begin{aligned}
    \pdv{\bm{v}'}{t}-\bar{\Omega}\pdv{\bm{v}'}{\phi} = \,&(\bm{v}'+r\bar{\Omega}\uvec_\phi)\cross\left(\curl{(\bm{v}'+r\bar{\Omega}\uvec_\phi})\right) - \grad\varphi \\
    =\,& \bm{v}'\cross\left(\curl{\bm{v}'}\right) + (r\bar{\Omega}\uvec_\phi)\cross(\curl{\bm{v}'}) \\
    &-\left(2\bar{\Omega}\uvec_z\cross\bm{v}'-\grad\bar{\Omega}^2r^2\right)- \grad\varphi.
    \label{eqn:euler_rot}
\end{aligned}
\end{equation}
It is useful to note that
\begin{equation}
\begin{aligned}
    \uvec_\phi\cross(\curl{\bm{v}'}) 
    &= \frac{1}{r}\grad(r v_\phi') - \frac{1}{r}\pdv{\bm{v}'}{\phi},
\end{aligned}
\end{equation}
so (\ref{eqn:euler_rot}) reduces to
\begin{equation}
\begin{aligned}
    \pdv{\bm{v}'}{t} =\,& \bm{v}'\cross\left(\curl{\bm{v}'}\right) + \grad(r\bar{\Omega}v_\phi'+\bar{\Omega}^2r^2-\varphi)-2\bar{\Omega}\uvec_z\cross\bm{v}'.
    \label{eqn:euler_rot_simp}
\end{aligned}
\end{equation}
Taking poloidal-toroidal projection of (\ref{eqn:euler_rot_simp}) gives
\begin{equation}
    \pdv{\bm{v}'}{t} = \mathbb{P}\left[\bm{v}'\cross\left(\curl{\bm{v}'}\right)-2\bar{\Omega}\uvec_z\cross\bm{v}'\right],
\end{equation}
so the transformed linear eigenvalue problem reads:
\begin{equation}
\begin{aligned}
    \sigma' \Tilde{\bm{R}}' =&\, \mathbb{M}_{mk} \Tilde{\bm{R}'} + \mathbb{P}_{mk}\big[-r\bar{\Omega}\uvec_\phi\cross(\curl{\mathbb{P}^{-1}[\bm{R}']})\\
    &+{\mathbb{P}^{-1}[\bm{R}']}\cross\left(\curl{(-r\bar{\Omega}\uvec_\phi)}\right)-2\bar{\Omega}\uvec_z\cross{\mathbb{P}^{-1}[\bm{R}']}\big] \\
    =&\, \mathbb{M}_{mk} \Tilde{\bm{R}'} + \mathbb{P}_{mk}\big[-r\bar{\Omega}\uvec_\phi\cross(\curl{\mathbb{P}^{-1}[\bm{R}']})\big]\\
    =&\, \mathbb{M}_{mk} \Tilde{\bm{R}'} +\iu m\bar{\Omega} \Tilde{\bm{R}'}, 
\end{aligned}
\end{equation}
which implies a shift in the wave frequency while the growth rate and the eigenvector are unchanged:
\begin{equation}
    \omega' = \omega + \iu m\bar{\Omega}
    \quad \mbox{and\ }\quad
    \Tilde{\bm{R}'} = \Tilde{\bm{R}}  .
    \label{eqn:translation1}
\end{equation}

\section{Numerical method}\label{appendix:mlegs}
The unbounded cylindrical domain is discretised using our previously developed spectral method based on mapped associated Legendre polynomials (MLEGS), as detailed in \citet{matsushima_marcus_1997} and \citet{lee_marcus}. The solver is available as an open-source software on Zenodo \citep{MLEGS}. The foundation of this method is a tunable algebraic mapping between the semi-infinite physical domain $r \in [0,\infty)$ and the finite Legendre interval $[-1,1]$. This mapping constructs Galerkin basis functions that are regular at the origin and retain appropriate algebraic decay in the far field. It also allows for redistribution of the radial collocation points, enabling high resolution in desired regions such as near the vortex core or critical layers.

The eigenvalue problem (\ref{eqn:evp}) is formulated based on the poloidal-toroidal decomposition of the velocity field, as described in \S\ref{sec:multi-scale}. This formulation inherently satisfies the solenoidal condition (mass conservation). As a result, for a spectral truncation of $M$ polynomials, the generalised eigenvalue problem possesses a complexity of $2M$, representing a significant reduction compared to primitive variable formulations. Additionally, the MLEGS basis functions are designed to be Galerkin at the boundaries, which naturally enforce the regularity conditions at these boundaries without the need for explicit boundary condition enforcement, resulting into a highly efficient numerical eigenvalue solver.

This high efficiency and tunable radial resolution are particularly crucial for treating solutions possessing critical layers. In the strict linear inviscid limit, the critical layer singularity classically relies on complex contour deformation for regularisation; in MLEGS, the spatial discretisation instead provides a numerical regularisation directly on the real axis, whose fidelity was systematically analysed by \citet{lee_marcus}. At finite Reynolds numbers, the same study demonstrated that clustering collocation points near the critical radius efficiently resolves the viscous critical layer, whose thickness follows the established $\Rey^{-1/3}$ law \citep{maslowe_1986, dizes_2004}. We emphasise that these results concern the faithful numerical representation of critical layer solutions; they do not by themselves establish that the singular eigenfunctions of the continuous spectrum are limits of weakly damped viscous modes (see \S\ref{sec:critical-layer}). Within this work, the numerical regularisation is exercised chiefly by the discrete critical layer modes, whose passive singularities are strongly suppressed, allowing us to solve the inviscid equations directly on the real axis just as rigorously as with complex contour deformation.

The method also has the ability to directly compute the quasi-steady response to external strain fields (zero-frequency limit of the dispersion relations), which captures how an external strain transmits into the vortex core \citep{moffatt_1994} and provides the necessary base flow perturbation to study elliptical instability. In particular, elliptical instability considers the perturbation with an assumed form:
\begin{equation}
\begin{aligned}
\bar{v}_r(r,\phi,z) &= \Tilde{\bar{v}}_r(r)\sin(2\phi) \\
\bar{v}_\phi(r,\phi,z) &= \Tilde{\bar{v}}_\phi(r)\cos(2\phi) \\
\bar{v}_z(r,\phi,z) &= \Tilde{\bar{v}}_z(r)\cos(2\phi) \\
\end{aligned}\, ,
\end{equation}
and requests the following boundary behaviours at far field:
\begin{equation}
\left.\begin{matrix}
\Tilde{\bar{v}}_r \rightarrow r
\\
\Tilde{\bar{v}}_\phi \rightarrow r
\\ 
\Tilde{\bar{v}}_z \rightarrow 0
\end{matrix}\right\}\,\text{as } r\rightarrow \infty\,.
\end{equation}
Despite the differences in form and boundary conditions to the eigenvalue problem (\ref{eqn:evp}), such perturbation is still an eigenmode of the unperturbed base flow. Given the assumed form of the disturbance field and $\bar{\sigma} = 0$, (\ref{eqn:howard-grupta}) can be reduced to a boundary-value problem in the absence of background axial flow, which reads
\begin{equation}
    r^2\dv[2]{g(r)}{r} + 5\,r\dv{g(r)}{r} - \frac{r^2\dv{}{r}{\zeta}}{V_\phi}g(r)
    = 0 \, ,
    \label{eqn:bvp}
\end{equation}
where the solution $g(r) \equiv \frac{1}{r^2}(r\Tilde{\bar{v}}_r)$ has boundary conditions $g(r)\rightarrow 1$ as $r\rightarrow\infty$ and $g(r)=\order{1}$ as $r\rightarrow 0$, and $V_\phi$ and $\zeta$ are the azimuthal velocity component and axial vorticity component of the base flow. 

\begin{figure}
    \centering
    \includegraphics[width=0.6\textwidth]{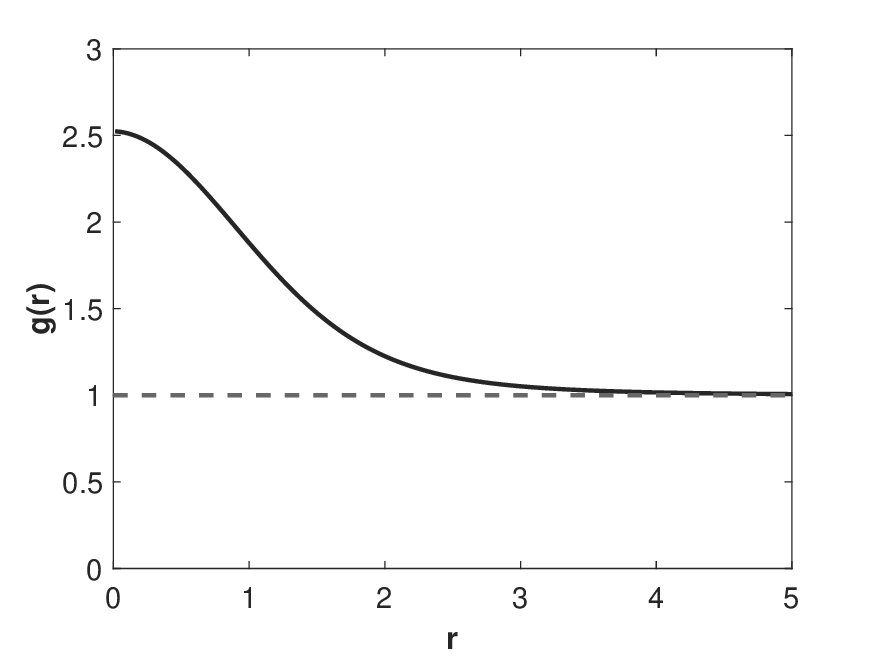}
    \caption{$g(r)$ as a function of $r$ with $s_0 \equiv g(0)= 2.524557$ and $g(r)\rightarrow 1$ at far field.}
    \label{fig:elliptical-strain}
\end{figure}

The same boundary-value problem has appeared in \citet{saffman_1975,dizes_1999,feys_maslowe_2016} and is typically solved via shooting methods \citep[see][]{feys_maslowe_2016}. Here, we note that the function $g(r)$ matches the boundary behaviours of the $0^{\mathrm{th}}$-order mapped associated Legendre polynomials, so (\ref{eqn:bvp}) can be solved using MLEGS directly. According to the boundary conditions, $g(r)$ can be represented as
\begin{equation}
    g(r) = \sum_{n=0}^\infty g^0_n P^0_{L_n} \, ,
\end{equation}
where $P^0_{L_n}$ is the $n^{\mathrm{th}}$-degree mapped associated Legendre polynomial of $0^{\mathrm{th}}$ order. We seek a spectral operator $\mathsfbi{G}$ such that,
\begin{equation}
\begin{aligned}
    \mathsfbi{G}\,[\bm{g}^0] &\equiv \mathsfbi{T}^{-1}\bigg[r^2\dv[2]{}{r}\bm{g} + 5\,r\dv[]{}{r}\bm{g} - \frac{r^2\dv{}{r}{\zeta}}{V_\phi}\bm{g}\bigg]\\
    &= \mathsfbi{T}^{-1}\bigg[r\dv{}{r}{\big(r(\dv{}{r}{\bm{g}})\big)} + 4\cdot r\dv{}{r}{\bm{g}}\bigg] - \mathsfbi{T}^{-1}\bigg[\frac{r^2\dv{}{r}{\zeta}}{V_\phi}\bm{g}\bigg]
    \label{eqn:discretized-bvp}
\end{aligned}
\end{equation}
where
\begin{equation}
    \bm{g} = \begin{vmatrix} g(r_0) \\ \vdots \\ g(r_M) \end{vmatrix} \, \text{and } \bm{g^0} =\begin{vmatrix} g^0_0 \\ \vdots \\ g^0_N \end{vmatrix}
\end{equation}
are vectors of the physical collocation points and the truncated spectral coefficients of $g(r)$ with a spectral projection operator such that $\bm{g} = \mathsfbi{T}[\bm{g}^0]$. The spectral operator for $r\dv{}{r}\{\cdot\}$, denoted as $\mathsfbi{D}$, is a band matrix operator as defined in \citet{matsushima_marcus_1997}, and the only remaining term in (\ref{eqn:discretized-bvp}) is the multiplication with $H(r) = {r^2 \dv{}{r} {\zeta}}/{V_\phi}$, which can be constructed in the physical space as a diagonal operator:
\begin{equation}
    \mathsfbi{H} \equiv 
    \begin{pmatrix}
    H(r_0) &  0 & \cdots\\ 
    0 & H(r_1)  & \cdots \\ 
    \vdots & \vdots & \ddots
    \end{pmatrix} \, .
\end{equation}
Hence,
\begin{equation}
    \mathsfbi{G} = 
    \mathsfbi{D}^2 
    + 4\,\mathsfbi{D}
    - \mathsfbi{T}^{-1}\mathsfbi{H}\mathsfbi{T}
\end{equation}
gives the matrix spectral operator $\mathsfbi{G}$ whose null space contains the solution of (\ref{eqn:bvp}), and $g(r)$ can be obtained by transforming the non-trivial zero-frequency eigenvector of $\mathsfbi{G}$ to the physical space.  The numerical solution $g(r)$ is plotted in figure~\ref{fig:elliptical-strain}, whose value at origin, $s_0 = 2.524557$, matches the numerical integration result computed by \citet{dizes_1999}.

\bibliographystyle{jfm}
\bibliography{main}

\end{document}